\documentclass[11pt,a4paper, amsmath,amssymb,nofootinbib,preprintnumbers, unsortedaddress]{article}
\pdfoutput=1
\usepackage{jheppub}
\usepackage{dsfont}
\usepackage{graphicx}
\usepackage{amssymb}
\usepackage{amsmath}
\usepackage{amsthm}
\usepackage{bbm}
\usepackage{mathrsfs}
\usepackage{slashed}
\usepackage{caption}
\usepackage{tikz,pgfplots,pgfmath,tikz-3dplot}
\usetikzlibrary{calc,shapes.geometric}
\tdplotsetmaincoords{80}{45}
\tdplotsetrotatedcoords{-90}{180}{-90}
\usepackage[dvipsnames]{xcolor}
\usepackage[framemethod=tikz]{mdframed}
\usepackage{comment}
\usepackage{esint}

\makeatletter
\renewcommand*\env@matrix[1][\arraystretch]{%
  \edef\arraystretch{#1}%
  \hskip -\arraycolsep
  \let\@ifnextchar\new@ifnextchar
  \array{*\c@MaxMatrixCols c}}
\makeatother

\newcommand\eq[1]{eq.~(\ref{eq:#1})}

\newcommand{\sn}[1]{\S~\ref{sec:#1}}
\newcommand{\fig}[1]{Fig.~\ref{fig:#1}}
\newcommand{\app}[1]{App.~\ref{app:#1}}

\newcommand\mf[1]{{\mathfrak{#1}}}
\newcommand\mc[1]{{\mathcal{#1}}}

\newcommand\CN{{\mc{N}}}
\newcommand\CM{{\mc{M}}}
\newcommand\CO{{\mc{O}}}

\def\shrug{\texttt{\raisebox{0.75em}{\char`\_}\char`\\\char`\_\kern-0.5ex(\kern-0.25ex\raisebox{0.25ex}{\rotatebox{45}{\raisebox{-.75ex}"\kern-1.5ex\rotatebox{-90})}}\kern-0.5ex)\kern-0.5ex\char`\_/\raisebox{0.75em}{\char`\_}}}

\definecolor{cardinal}{rgb}{0.6,0,0}
\definecolor{darkgreen}{rgb}{0,0.5,0}
\definecolor{golden}{rgb}{0.92, 0.7, 0}
\definecolor{midnight}{rgb}{0, 0, 0.5}
\definecolor{darkblue}{rgb}{0.2, 0, 0.8}

\newtheorem{theorem}{Theorem}[section]
\newtheorem{lemma}[theorem]{Lemma}

\begin{document}

\title{Conformal Defects in Neural Network Field Theories}

\author[a]{Pietro Capuozzo,}
\author[b]{Brandon Robinson}
\author[a,c]{and Benjamin Suzzoni}

\affiliation[a]{STAG Research Centre, University of Southampton,\\ Southampton, SO17 1BJ, UK}
\affiliation[b]{Institute of Physics, University of Amsterdam,\\ 904 Science Park, 1098 XH Amsterdam, Netherlands}
\affiliation[c]{Department of Mathematical Sciences, Ulsan National Institute of Science and Technology,\\ 50 UNIST-gil, Ulsan 44919, South Korea}

\emailAdd{P.Capuozzo@soton.ac.uk}
\emailAdd{b.j.robinson@uva.nl}
\emailAdd{b.suzzoni@benterre.com}

\abstract{Neural Network Field Theories (NN-FTs) represent a novel construction of arbitrary field theories, including those of conformal fields, through the specification of the network architecture and prior distribution for the network parameters. In this work, we present a formalism for the construction of conformally invariant defects in these NN-FTs.  We demonstrate this new formalism in two toy models of NN scalar field theories. We develop an NN interpretation of an expansion akin to the defect OPE in two-point correlation functions in these models.}

\arxivnumber{2512.07946}

\maketitle
\tableofcontents

\section{Introduction}\label{sec:introduction}

Conformal Field Theories (CFTs) occupy a distinguished place within the space of Quantum Field Theories (QFTs). They are characterized by an enhanced set of (global) spacetime symmetries, furnished by the extension of the Poincar\'e group to include scaling and special conformal transformations. They are concretely defined in that they lack an intrinsic energy scale, and thus require no such information in their specification. CFTs are often found at fixed points of Renormalization Group (RG) flows, and are therefore important elements in the study of the behavior of QFTs with relevant couplings. CFTs also exist in arbitrary dimensions, and --- in certain, often supersymmetric, cases --- can be related across dimensions via compactifications and/or RG flows. Furthermore, unitarity is not a strictly necessary condition for defining a CFT; non-unitary CFTs find application in modeling numerous physical systems and have rich spaces of conformal boundary/defect conditions, see e.g. \cite{Chalabi:2022qit,Herzog:2024zxm} for a recent discussion. Because of their ubiquity, CFTs play a crucial role in the study of many areas of physics, from critical phenomena in condensed matter systems to string/M-theory and mathematics; hence, new methods and insights in the understanding of CFTs can often impact the study of disparate fields in surprising ways.

A novel method for defining a (C)FT comes from the study of Neural Networks (NNs) --- so-called NN-FTs \cite{Halverson:2020trp, Halverson:2021aot} --- which offers an interesting bridge between a mature domain of theoretical physics and the rapidly developing field of Machine Learning (ML), and opens up a number of new avenues for cross-pollination between the two. Starting from the observation that the Gaussian Processes (GPs) observed in certain asymptotic limits of NNs, e.g. that of infinite width, can be brought into correspondence with free field theories, one can then understand any non-Gaussianities as arising from interactions in the corresponding QFT \cite{Demirtas:2023fir}. Moreover, one can encode global and gauge symmetries into appropriate equivariant structures within the NN \cite{Maiti:2021fpy}, and  the interpretation of Poincar\'e symmetries as symmetries of the NN's input space has recently been extended to include global conformal symmetries \cite{Halverson:2024axc}. The ability of fairly simple network architectures to realize the $d$-dimensional Euclidean conformal symmetry group $SO(d+1,1)$ (or, in Lorentzian signature, $SO(d,2)$) and be constrained by its kinematics is quite remarkable. Given the prevalence and importance of CFTs in fields such as string/M-theory, it is crucial to build out the NN-FT program and explore the limits of the paradigm's ability to construct new classes of CFTs. 

Boundaries, impurities, and defects of arbitrary co-dimension have long played a central role in the study of CFTs. A geometric picture of defects begins with a reference `ambient' space on which a CFT is defined, and into which a lower-dimensional `defect' submanifold is embedded. This submanifold supports degrees of freedom that are restricted to it, and which are typically coupled to degrees of freedom in the ambient CFT. These extended objects represent a crucial facet not only of the abstract definition of a QFT, but also of concrete models of physical systems, e.g. the Kondo effect and quantum Hall physics, see a recent review in \cite{Andrei:2018die}. Thus, given this novel method of defining conformal fields via NNs, a natural question to ask is how such a construction can accommodate the full spectrum of a CFT, including, in particular, any allowed non-local deformations it may support. The definition of conformal fields through NNs relies on concepts in CFTs such as the embedding space formalism, and the incorporation of conformal symmetry-breaking deformations into this framework is well understood \cite{Billo:2016cpy, Lauria:2017wav}. With that in mind, the purpose of the following work is to unite these approaches and develop a method for constructing novel conformal defects in the context of NN-FTs. 
  
While NNs offer a powerful framework for constructing new classes of field theories, including those with conformal defects, the concept of defects is itself of interest from an ML perspective. Indeed, drawing on lessons from field theory, defects offer a new way to introduce non-Gaussianities by symmetry breaking alone.  As a simple example, consider a GP with an arbitrary $d$-dimensional input space --- in the QFT sense, a free field theory in $d$ dimensions. Defining the theory instead on a background with a boundary, we must specify the behavior of fields near the locus of a codimension-one spacetime submanifold, which in turn determines the couplings to any boundary degrees of freedom. This interaction between bulk and boundary fields naturally gives rise to non-Gaussianities. 
  
This example of a free field theory in the presence of a boundary can be understood in the NN context through a reduced symmetry in the NN parameters, collectively referred to by $\Theta$ throughout the body of the text, whose distribution is $P(\Theta)$. Consider a conformally invariant $d$-dimensional measure $\mathrm{d}P(\Theta)$ in parameter space. Breaking this symmetry down to a $(d-1)$-dimensional conformal subgroup can be achieved, for instance, by breaking the translational symmetry along one direction in $\Theta$-space, e.g. by freezing the value of one of the parameters, $\Theta_\perp = c$. Indeed, this breaks the symmetry in the $\Theta_\perp$ direction while preserving it along the remaining $\Theta\setminus\{\Theta_\perp\}$.  This logic can be extended to the breaking of other kinds of symmetries, e.g. rotations, and to higher co-dimension defects, in a straightforward fashion.
  
There is an alternative and particularly intriguing way to think of this symmetry breaking effect. In the NN-FT picture, the input space for the NN is regarded as the background space(-time) on which the field theory is defined; here, for simplicity, we will take this to be $\mathbb{R}^d$. The introduction of a co-dimension $q$ defect from the input space perspective has conceptual overlap with the so-called `manifold hypothesis', which states that data is typically distributed on a submanifold immersed\footnote{Note that it is not guaranteed that the intrinsic topology of the `data submanifold' coincides with the (pullback) subset topology in the ambient space. Hence, it is not always the case that there is a smooth embedding.} in a higher-dimensional input space; see \cite{fefferman2016testing} for an introduction to the topic. Thus, one could ask how out-of-distribution or off-of-defect (OOD) samples, i.e. those lying away from the support of the true data distribution, can be used to probe the data submanifold itself. That is, how do expectations of OOD samples with respect to the joint distribution formed by the parameters spanning the directions along the submanifold and the directions normal to it characterize the submanifold structure? Addressing these questions requires methods familiar from the study of correlation functions in defect CFTs.   

One of the defect CFT techniques that we will employ to gain deeper insight into the structure of correlation functions for these novel conformal defects from NNs is the decomposition of `ambient fields' in the spectrum of `defect fields', which draws from the defect operator product expansion (OPE). The lack of an NN-FT operator formalism notwithstanding, the defect OPE is a useful analogy, as it allows the correlation functions of fields `sourced' in the ambient space to be expressed in terms of defect degrees of freedom transforming in specific representations of the lower-dimensional conformal symmetry group, which depending on the embedding may or may not include rotational symmetries in the directions normal to the submanifold. By organizing correlation functions in terms of their quantum numbers under the defect symmetries, we will be able to reconstruct correlation functions of ambient fields in higher-dimensional input spaces from linear combinations of expectations taken with respect to smaller (i.e. defect) networks with specific weights. 
 
In the present work, we propose a methodology to construct conformal symmetry breaking layers in NN-FTs and thereby incorporate conformal defects into the general paradigm.  We briefly review the necessary components of building conformal fields from NNs and how defects fit into the embedding space formalism in \sn{NN-CFT}. In \sn{defect-NNs}, we synthesize these concepts into a framework for building conformal defects of arbitrary codimension from NNs. With this formalism in hand, we explore two classes of toy models to elucidate its utility in \sn{examples}. We then recapitulate our findings and discuss future directions in \sn{conclusions}. In \app{two-point-fns}, we collect intermediate results that are useful for computing two-point functions in the toy models that we consider.

\section{Background}\label{sec:NN-CFT}
In this section, we  briefly review the concepts necessary to construct conformal fields in the NN-FT paradigm and their extension to conformal defects, as well as set the notation used in the later sections.  We begin by covering the basic notions of the embedding space formalism and its utility in building conformal fields from neural networks \cite{Halverson:2020trp}.  We then summarize how conformal defects are encoded in the embedding space framework \cite{Billo:2016cpy}.

\subsection{NN-FTs for conformal fields}\label{sec:embedding-space}

To begin, let us recall the basic logic of the NN-FT formalism \cite{Halverson:2020trp}.  Central to both the ML and FT descriptions are random functions, i.e. function-valued random variables. On the ML side, a choice of network architecture defines a family of functions $\{\Phi:\mathcal{X}\to\mathbb{R}\mid\Theta\in\mathcal{P}\}$ over the input space $\mathcal{X}$, where the parameters $\Theta$ undergo a random initialization $\Theta\sim P$, for a probability measure $P$ over the parameter space $\mathcal{P}$. Hence, under random initialization, $\Phi(X)$ defines a random function. We will henceforth use the simpler notation $\Phi(X)$ and leave the $\Theta$-dependence implicit. In a Euclidean (Q)FT, fields $\Phi(X)$ are modeled as random distributions from a probability space to a configuration space, which is typically the dual of Schwartz space over the Euclidean spacetime. In such theories, the action functional $S[\Phi]$ constructs a Gibbs-type probability measure $\mathcal{D}\Phi~e^{-S[\Phi]}$, where $\mathcal{D}\Phi$ denotes the path integral measure, a heuristic analogue of the Lebesgue measure on the infinite-dimensional space of field configurations. Given this probabilistic model for the field, and given an external source $J(X)$ --- a Schwartz test function --- one defines the generating functional $Z[J]$ for correlation functions as a path integral over field configurations,
\begin{align}
    Z[J] = \int \mathcal{D}\Phi~ e^{-S[\Phi]+\int \mathrm{d}^dX~J(X)\Phi(X)}~.
\end{align}
This notion also extends to Lorentzian QFTs satisfying the Wightman axioms \cite{wightman}, via analytic continuation in time from Osterwalder-Schrader Euclidean theories \cite{Osterwalder:1973dx}. A correspondence can be established between the field $\Phi(X)$ of a Euclidean field theory at a point $X$ and the output of a neural network given an input $X$, which, as above, we also denote $\Phi(X)$; the path integral $Z[J]$ can then be interpreted as an expectation value over network realizations,
\begin{align}
    Z[J] &= \int \mathrm{d}\Theta~P(\Theta)~ e^{\int \mathrm{d}^dX~J(X)\Phi(X)}\nonumber\\
    &=\mathbb{E}_{\Theta\sim P}\left[e^{\int \mathrm{d}^dX~J(X)\Phi(X)}\right]~,
\end{align}
which serves as a moment generating functional (MGF) of the random function $\Phi$. In this formulation, the QFT path integral and the neural network ensemble average play analogous roles. In particular, the action functional $S[\Phi]$ which weighs field configurations in the QFT finds its formal analogue, on the NN side, in the initialization law $P$ that weighs the network parameters.

Taking the definition of a field within the NN-FT paradigm to be the specification of the network architecture --- e.g. in the perceptron $\Phi(\vec{X}) = h( \vec{w}\cdot \vec{X} + \vec{b})$, the specification of weights $w_i$, biases $b_i$, and non-linearities $h(~\cdot~)$ --- then it is clear that the correlation functions of the FT can be computed from the moments of $P(\Theta)$ as follows,
\begin{align}
    \left<\Phi(X_1)\ldots\Phi(X_n)\right>
    &=\int\mathrm{d}\Theta~ P(\Theta)~\Phi(X_1)\cdots \Phi(X_n)\nonumber\\
    &= \mathbb{E}_{\Theta\sim P}\left[\Phi(X_1)\cdots\Phi(X_n)\right]~.
\end{align}

If one wishes to consider this new approach to defining an FT through its description in the parameter space $\mathcal{P}$, then it is natural to inquire about the ways in which symmetries are encoded in an NN. As an example, let us consider the perceptron model as a scalar field under the action of translations. The field theoretic translation $X\to X+a$ of the background space(time) coordinates, for $a\in\mathbb{R}$, is understood in the NN picture as a translation of the network's inputs; therefore, demanding translational invariance requires that the biases transform as $\vec{b}\to \vec{b}^\prime= \vec{b}+ \vec{w}\cdot\vec{a}$, so that $\mathrm{d}\Theta~P(\Theta)\to \mathrm{d}\Theta^\prime~P(\Theta^\prime)$. Hence, in the NN parameter space description, symmetry transformations of the inputs $X\to X^\prime$ are realized through the capacity of the NN parameters to compensate transformations $\Theta\to\Theta^\prime$ in such a way as to leave the expectations $\mathbb{E}_{\Theta\sim P(\Theta)}[\Phi(X_1)\ldots]$ invariant. The logic can be extended from the simple example of translations to include all of the Poincar\'e symmetries \cite{Halverson:2020trp,Halverson:2021aot, Maiti:2021fpy} and conformal symmetries \cite{Halverson:2024axc}.  

As a brief aside, we note that global symmetries act differently in the context of NN-FTs since they are transformations in field (or network) space rather than in the input or parameter spaces. For instance, if we consider the analog of an $O(N)$ model of NN-FTs, we would have a collection of architectures $\Phi_a(X)$ that are rotated into one another, $\Phi_b(X) = R_b{}^a\Phi_a(X)$ by $R\in O(N)$. The concept of global symmetries occurs naturally in the context of defects in NN-FTs, but its appearance is incidental to breakings of the symmetries of the input space which preserve a global subgroup of normal bundle symmetries; we will momentarily forgo further discussion on this topic, only to return to it in \sn{conclusions}. 

Let us now return to the realization of conformal fields in NN-FTs. An obstruction to simply extending the logic of encoding Poincar\'e symmetries to global conformal symmetries is that, for a given CFT defined on a flat $d$-dimensional manifold, the Euclidean (Lorentzian, resp.) global conformal symmetry group $SO(d+1,1)$ ($SO(d,2)$, resp.) acts {\emph{non-linearly}} on fields. Fortunately, it is possible to bypass this issue by extending the $d$-dimensional (Euclidean) background to $d+2$ dimensions by adding one spacelike and one timelike direction. In the resulting {\emph{embedding space}}, the $d$-dimensional conformal symmetries act linearly, as they are induced by the Lorentz symmetries of the $(d+2)$-dimensional background \cite{Dirac:1936fq}. Within this higher-dimensional space, one identifies its null cone and quotients it by non-zero rescalings. The $d$-dimensional CFT is defined on the resulting space, which is referred to as the projective null cone ($\mathbb{P}$NC). The Poincar\'e section (PS) selects a unique representative from each projective equivalence class, i.e., it embeds $\mathbb{R}^d\hookrightarrow\mathbb{P}$NC. As reviewed below, the embedding space formalism offers a straightforward method to compute correlation functions in CFTs.

Let us make this discussion somewhat more concrete. Consider a CFT defined on a flat $d$-dimensional background $\CM = \mathbb{R}^d$ with coordinates $\{x^\mu\}_{\mu=1}^d$ and endowed with the Euclidean metric producing the line element $\mathrm{d}s^2 = \delta_{\mu\nu}\mathrm{d}x^\mu\mathrm{d}x^\nu$. Its embedding space $\mathbb{R}^{d+1,1}$ has coordinates $\{X^M\}_{M=1}^{d+2}$ and is endowed with the Minkowski metric with line element $\mathrm{d}s^2= \eta_{MN}\mathrm{d}X^M\mathrm{d}X^N$, which we take to have mostly plus signature $(++\ldots +-+)$. We will mostly employ lightcone coordinates $X = (X^+,X^-,  x^\mu)$, where $X^\pm =X^{d+2}\pm X^{d+1} $. The $\mathbb{P}$NC is then the locus defined by the projective condition $X\sim \lambda X$, for $\lambda \in \mathbb{R}\setminus\{0\}$, and  by the null condition $X^2\equiv X\cdot X=0$, where $\cdot$ denotes the invariant scalar product with respect to the metric $\eta_{MN}$. The physical theory is recovered on the PS defined by $X\mapsto (1,x^2, x^\mu)$. This construction is visualized in \fig{nullcone}.

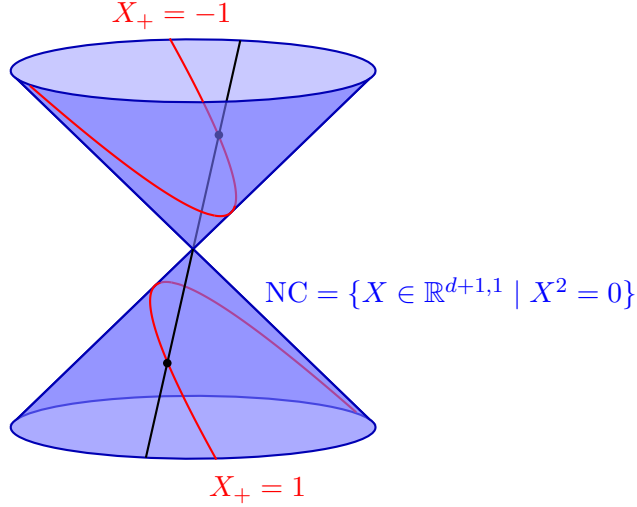
\begin{figure}[htbp]
  \centering
    \begin{tikzpicture}[tdplot_main_coords, scale=0.8]
    \coordinate (O) at (0,0,0);
    \draw[canvas is xy plane at z=-3, draw=blue!70!black, thick, fill=blue!35!white, fill opacity=.6] (60:3) arc (60:210:3) -- (O) -- cycle;
    \draw[canvas is xy plane at z=3, draw=blue!70!black, thick, fill=blue!35!white, fill opacity=.6] (30:3) arc (30:240:3) -- (O) -- cycle;
    \draw[red, thick]
    plot[domain=0.7:3, samples=200, variable=\t]
    ({sqrt(2*\t - 1)}, {1 - \t}, {\t});
    \draw[red, thick]
    plot[domain=-0.5:-2.7, samples=200, variable=\t]
    ({sqrt(-2*\t - 1)}, {-1 - \t}, {\t});
    \draw[red, thick]
    plot[domain=-0.5:-0.7, samples=200, variable=\t]
    ({-sqrt(-2*\t - 1)}, {-1 - \t}, {\t});
    \draw[black, thick] plot[domain=-0.7:0, samples=5, variable=\t]
    ({1.5*\t}, {-2.598*\t}, {-3*\t});
    \draw[black, thick] plot[domain=-1:-0.7, samples=5, variable=\t]
    ({1.5*\t}, {-2.598*\t}, {-3*\t});
    \fill[fill=black] (-7.4,8) circle (2pt);
    \draw[canvas is xy plane at z=-3, draw=blue!70!black, thick, fill=blue!55!white, fill opacity=.6] (60:3) arc (60:-150:3) -- (O) -- cycle;
    \draw[canvas is xy plane at z=3, draw=blue!70!black, thick,  fill=blue!55!white, fill opacity=.6] (30:3) arc (30:-120:3) -- (O) -- cycle;
    \draw[red, thick]
    plot[domain=3:4.3, samples=200, variable=\t]
    ({sqrt(2*\t - 1)}, {1 - \t}, {\t});
    \draw[red, thick]
    plot[domain=0.5:0.8, samples=200, variable=\t]
    ({sqrt(2*\t - 1)}, {1 - \t}, {\t});
    \draw[red, thick]
    plot[domain=0.5:2.7, samples=200, variable=\t]
    ({-sqrt(2*\t - 1)}, {1 - \t}, {\t});
    \draw[red, thick]
    plot[domain=-0.55:-4.3, samples=200, variable=\t]
    ({-sqrt(-2*\t - 1)}, {-1 - \t}, {\t});
    \draw[black, thick] plot[domain=0:1, samples=5, variable=\t]
    ({1.5*\t}, {-2.598*\t}, {-3*\t});
    \fill[fill=black] (7.4,-8) circle (2pt);
    \node[blue] at (6,0) {$\mathrm{NC}=\{X\in\mathbb{R}^{d+1,1}\mid X^2=0\}$};
    \node[red] at (17,-15.5) {$X_+=1$};
    \node[red] at (-16,15.5) {$X_+=-1$};
    \node[blue] at (-8,0) {};
  \end{tikzpicture}
  \caption{A visual representation of the null cone $\text{NC}\subset\mathbb{R}^{d+1,1}$. The Poincar\'e section (PS) $X_+=1$ and its antipodal locus $X_+=-1$ are shown in red; they select a unique representative from each projective orbit $X\sim\lambda X$ within each section. An example of such a null ray is shown in black.}
  \label{fig:nullcone}
\end{figure}

As mentioned above, the generators of the conformal symmetries on the PS are induced by the Lorentz symmetry generators of the embedding space, which are of the form $J_{MN} = X_M \partial_N -X_N\partial_M$. Pulling back to the PS, spatial rotations $L_{ij}$ correspond to the transformations which stabilize points on the $\mathbb{P}$NC, i.e. the subgroup of the full stabilizer group preserving the projective and null conditions, which is generated by $J_{ij}$. Translations $P_\mu$ correspond to boosts along $X_+$, which are generated by $J_{\mu+}$, while the special conformal transformations $K_\mu$ are induced by the generator $J_{\mu-}$. Lastly, dilatations $D$, i.e. scale transformations, originate from $J_{+-}$. 

The extended background space(time) (i.e. input space) symmetries impose strong constraints on observables such as the connected correlation functions of conformal fields, or rather expectations of the form $\mathbb{E}[\Phi(X_1)\ldots \Phi(X_n)]$.  In particular, in a $d$-dimensional CFT, conformal symmetry is constraining enough that the two-point function of scalar fields $\Phi_a$ and $\Phi_b$, with respective scaling dimension (i.e. eigenvalue under dilatations) $\Delta_a$ and $\Delta_b$, is uniquely determined\footnote{Here we only write the contribution from spinless, primary operators within the OPE; where by `primary' we are referring to the state created by acting with the operator $\CO$ on the conformal vacuum that is annihilated by generators of special conformal transformations, i.e. $K_\mu|\CO\rangle=0$. In a general CFT, there will be contributions from spin-$\ell$ primaries alongside the entire tower of descendants. For the purposes of this lightning review, we will make similar truncations for the three- and four-point functions. See \cite{Dolan:2000ut} for a thorough discussion.} to be
\begin{align}
    \left<\Phi_a(x_1)\Phi_b(x_2)\right> = \frac{\delta_{ab}}{|x_1 - x_2|^{2\Delta_a}}\propto \frac{\delta_{ab}}{(X_1\cdot X_2)^{\Delta_a}}~,
\end{align}
where the last relation is an equality only up to numerical factors. Similarly, the three-point function of scalar fields in a CFT is also highly constrained, but it nevertheless contains important physical information. For three scalar fields $\Phi_a,\Phi_b,\Phi_c$, conformal symmetry fixes the following form,  
\begin{align}
    \left<\Phi_a(x_1)\Phi_b(x_2)\Phi_c(x_3)\right> = \frac{C_{abc}}{|x_1-x_2|^{\Delta_{a}+\Delta_b-\Delta_c} |x_2-x_3|^{\Delta_b+\Delta_c-\Delta_a} |x_1-x_3|^{\Delta_a+\Delta_c-\Delta_b}}~.
\end{align}
The constant $C_{abc}$ appearing in the three-point function is the operator product expansion (OPE) coefficient. Higher-point functions are less constrained, in that they can be determined by kinematics only up to functions of the conformal cross ratios
\begin{align}
\begin{split}
    u &= \frac{(X_1\cdot X_2)(X_3\cdot X_4)}{(X_1\cdot X_3)(X_2\cdot X_4)} \overset{\rm{P.S.}}{=} \frac{|x_1-x_2|^2|x_3-x_4|^2}{|x_1-x_3|^2|x_2-x_4|^2},\\
    v &= \frac{(X_1\cdot X_4)(X_2\cdot X_3)}{(X_1\cdot X_3)(X_2\cdot X_4)} \overset{\rm{P.S.}}{=} \frac{|x_1-x_3|^2|x_2-x_3|^2}{|x_1-x_3|^2|x_2-x_4|^2}.
    \end{split}
\end{align}
where $\overset{\rm{P.S.}}{=}$ indicates restriction to the Poincar\'e section.  For example, the four-point function of scalars $\Phi_a,\ldots,\Phi_d$ in the scalar channel is \cite{Dolan:2000ut}
\begin{align}
    \left<\Phi_a(x_1)\Phi_b(x_2) \Phi_c(x_3)\Phi_d(x_4)\right> =\frac{\left(\frac{|x_2-x_4|}{|x_1-x_4|}
    \right)^{\Delta_a-\Delta_b}\left(\frac{|x_1-x_4|}{|x_1-x_3|}\right)^{\Delta_c-\Delta_d}}{|x_1-x_2|^{\Delta_a+\Delta_b}|x_3-x_4|^{\Delta_c+\Delta_d}}g(u,v)~.
\end{align}
Typically, in CFTs, one must appeal to the OPE in order to find expressions for $g(u,v)$ that can be constrained by the `crossing symmetry' reflecting the equivalent ways of applying the OPE pairwise on the fields.  The reduced functions that appear in the OPE expansion and that are related to the function $g(u,v)$ appearing above are called \emph{conformal blocks}. Remarkably, in NN-FTs there are simple toy models in which the function $g(u,v)$, and hence the conformal blocks, can be directly computed \cite{Halverson:2024axc}.

Let us pause here for a moment to consider in greater detail the appearance of $C_{abc}$ and the relation of the four-point function $g(u,v)$ to the OPE through its expansion in conformal blocks; let us think about the implications for importing the OPE language into the context of NN-FTs.  In ordinary QFTs and CFTs, the OPE is a well-defined concept and plays a vital role in the study of the spectra of such theories.  As a quick reminder, for two CFT operators $\CO_a$ of scaling dimension $\Delta_a$, one inserted at the origin and the other at a distance $x_1$ away, and provided there are no other operator insertions within a ball of radius $x_1$, there is an equivalent expression in terms of an expansion in primary operators $\CO$ of scaling dimension $\Delta$ and the tower of their descendants
\begin{align}\label{eq:cft-ope}
    \CO_1(x_1)\CO_2(0)\sim \sum_{\CO}C_{12\CO}(x_1,\partial_x) \CO(x)\Big|_{x\to0}~.
\end{align}
The precise form of the coefficient $C_{abc}$, namely its dependence on the insertion point $x_1$ and the raising operators $P_\mu$, is fixed by conformal symmetry.  As an operator statement, the OPE can be employed in correlation functions, which allows one to leverage analytic expressions for lower-rank correlation functions as constraints; whence the conformal bootstrap program.  In the context of this brief review, it is then clear how the form of the three-point function quoted above arises and the role of $C_{abc}$.  

The consequence of using this OPE language in the context of correlators of conformal fields will be explored in greater detail below, but the expression in \eq{cft-ope} opens an interesting line of thought. The correlation of the outputs of two network realizations acting on similar regions of input space --- which in the NN-FT paradigm is understood as two field insertions at nearby points in the background space(time) --- is equivalent, in this OPE-like expansion, to a (weighted) sum over the tower of expectations of the outputs of independent networks.

\subsection{Conformal defects in embedding space}\label{sec:review-defect-embedding-space}

The insertion of a (flat) $p$-dimensional conformal defect in the embedding space formalism amounts to the (flat space) symmetry breaking $SO(d+1,1)\to SO(p+1,1)\times SO(q)_N$ and the splitting $T\CM\to T\mf{M}_p \oplus N\mf{M}_p$ of the tangent bundle $T\CM$ of the ambient space into tangent and normal bundles to the defect's worldvolume $\mf{M}_p$. Lifting to the projective null cone $\mathbb{P}\mathrm{NC}\subset\mathbb{R}^{d+1,1}$, the embedding space coordinates $\{X^M\}$ now split into two sets:
\begin{align}
    \{X^M\} = \{\widehat{X}^A\}\sqcup\{\widetilde{X}^I\}~,
\end{align}
where the coordinates $\widehat{X}^A$ are tangential to the uplifted defect worldvolume, while the coordinates $\{\widetilde{X}^I\}$ are normal to it. Furthermore, the breaking of the full global conformal symmetry group implies that the quadratic form $X\cdot Y$ is no longer a scalar under the reduced symmetry group. Instead, due to the orthogonality of the decomposition, it splits into the sum of an $SO(p+1,1)$-invariant product, denoted by $\bullet$, and an $SO(q)_N$-invariant one, denoted by $\circ$, as follows:
\begin{align}\label{eq:dot-decomposition}
    X\cdot Y = X\bullet Y + X\circ Y~.
\end{align}

Given the breaking of the global conformal symmetry group to its defect subgroup and the accompanying decomposition of the scalar product into tangential and normal components, one can expect novel features to arise in defect CFTs compared to stand-alone CFTs. Firstly, one can now distinguish between degrees of freedom that live in the higher-dimensional ambient space $\CM$ and those constrained to the defect submanifold $\mf{M}_p$, which also gives rise to new ambient-to-defect couplings. Secondly, the constraints on correlation functions can change drastically. For example, in an ordinary $d$-dimensional CFT, a scalar primary $\CO(x)$ cannot have a non-trivial one-point function $\langle \CO(x)\rangle$, as this would violate the conformal symmetry. However, in a defect CFT, the broken symmetries allow a scalar primary inserted in the ambient space at a distance $x_\perp$ in the normal directions away from the defect submanifold to have a non-trivial one-point function, albeit this is determined up to an overall constant,
\begin{align}\label{eq:defect-one-pt-fn-1}
    \langle\CO_\Delta(x_\perp)\rangle \sim \frac{\widehat{b}_\CO}{\Vert X \Vert_{\circ}^{\Delta}} \sim \frac{\widehat{b}_\CO}{|x_\perp|^{\Delta}}~,
\end{align}
where we introduced the shorthand $\Vert X \Vert_{\circ}:=\sqrt{X\circ X}$.  The constant $\widehat{b}_\CO$ is related to the coupling of the ambient operator $\CO$ to the defect degrees of freedom.  

As in ordinary CFTs, it is useful to express defect correlation functions in terms of defect conformal cross ratios. Since there are now mixed correlators with arbitrary numbers of ambient and defect fields, the space of defect cross ratios is much larger \cite{Billo:2016cpy, Gadde:2016fbj, Lauria:2017wav}. In the toy models discussed below, we will only compute correlation functions with up to two (ambient or defect) fields, and thus it will suffice to define the common cross ratios 
\begin{align}\label{eq:cross-ratios}
    \chi = -\frac{2X_1\bullet X_2}{\Vert X_1\Vert_{\circ} ~\Vert X_2\Vert_{\circ} }~,\qquad \cos\psi = \frac{X_1\circ X_2}{\Vert X_1\Vert_{\circ} ~\Vert X_2\Vert_{\circ} }~.
\end{align}

It can be surmised from the appearance of the ambient-to-defect coupling in the one-point function of \eq{defect-one-pt-fn-1} that there is a defect version of the OPE that expands ambient operators in terms of defect primaries and their towers of descendants.\footnote{Though it may seem odd at first glance that a single operator admits an OPE, one can utilize a sort of `method of images' for defect CFTs and rationalize the defect OPE of a single ambient insertion through its interaction with its `image' operator constrained by the particular conditions placed on fields in the limit of restriction to the defect submanifold \cite{Nishioka:2022ook}. A useful mnemonic is that for ambient insertions away from the defect there is a sort of `doubling' by the image operator, and hence, one-point functions of ambient operators have structures similar to those of two-point functions in ordinary CFTs. Similarly, two-point functions of ambient operators are reminiscent of four-point functions in ordinary CFTs, and so on.} Explicitly, for a scalar primary operator $\CO_\Delta$ in a unitary defect CFT with scaling dimension $\Delta >0$, we can expand an ambient insertion in terms of the defect spectrum as \cite{Billo:2016cpy, Lauria:2017wav}
\begin{align}\label{eq:defect-ope}
    \CO_\Delta = \sum_{\widehat{\CO},s}\sum_{n}(-1)^n\frac{\widehat{b}_{\CO\widehat{\CO}}~\Gamma(\widehat\Delta+1-\frac{p}{2})}{2^{2n}\Vert X \Vert_{\circ}^{\Delta-\widehat{\Delta}-2n}~\Gamma(n+1)\Gamma(\widehat\Delta +1+n-\frac{p}{2})}(\widehat\nabla^2)^n(w \circ \tilde{\nu})^{s}\widehat{\CO}_{\widehat{\Delta},s}~,
\end{align}
where $\widehat{\nabla}^2$ is the Laplacian along the defect submanifold, $b_{\CO\widehat{\CO}}$ is the ambient-defect coupling, $s$ is the $SO(q)_N$ quantum number, $w$ is a transverse polarization vector 
\begin{align}
    \widehat{\CO}_{\widehat{\Delta},s} = w^{i_1}\ldots w^{i_s}\widehat{\CO}_{\Delta;s,i_1\ldots,i_s}
\end{align}
 for a defect operator $\widehat{\CO}_{\widehat{\Delta},s}$ with transverse spin $s$ satisfying $w\circ w =0$, and $\widetilde\nu^I = X^I/\Vert X\Vert$ is a unit normal to the defect submanifold. The sum over $n>0$ spans the tower of descendants from the defect primary at $n=0$, i.e. $\widehat{\CO}_{\widehat{\Delta},s}$. Looking back to \eq{defect-one-pt-fn-1}, the overall coefficient can be identified with the coupling to the defect identity operator, $\widehat{b}_\CO\equiv \widehat{b}_{\CO\mathds{1}}$.

\section{Conformal defects in NN-FTs}\label{sec:defect-NNs}

In this section, we present a method to build NNs that realize defect conformal symmetry. This construction is inspired by certain approaches that formulate defect CFTs as deformations of theories with larger ambient global conformal symmetry groups. That is, we begin with a NN-FT that realizes a global $SO(d+1,1)$ symmetry. In the parameter space language reviewed above, this implies considering a probability distribution $P(\Theta)$ which is invariant under $d$-dimensional rotations, translations, scaling, and special conformal transformations. In the input space, which we take to be $M = \mathbb{R}^d$, we can think of this more naturally in the FT language, wherein the conformal transformations of a field $\Phi(X)$ correspond to spacetime symmetries. Hence, our objective within this section is to formalize the embedding of a $p$-dimensional submanifold $\mathfrak{M}_p\hookrightarrow M$ in the input and parameter space pictures of an NN-FT; this will be accomplished by modifying the architecture and parameter distribution in such a way as to realize at least an $SO(p+1,1) \subset SO(d+1,1)$ symmetry. We will see that generic choices that break the higher-dimensional conformal symmetries do not lead to features that typically characterize conformal defects, such as non-trivial one-point functions for `scalar primaries' and defect rotational symmetry, but it is nevertheless possible to judiciously choose an architecture and a factorization of $P(\Theta)$ that manifest the maximal conformal subgroup $SO(p+1,1)\times SO(d-p)_N\subset SO(d+1,1)$.

Consider a conformal field $\Phi(X)$ constructed from an NN acting on an input space spanned by the restriction of the embedding space coordinates $\{X\}$ to a physical point on the Poincar\'e section.  As reviewed in \sn{embedding-space}, one of the constraints imposed on the choice of architecture by the presence of conformal symmetry\footnote{While it is not a question that we directly address in the current work, should we allow for the architecture $\Phi_{s_1\ldots s_J}(X)$ to transform in a non-trivial, symmetric traceless, spin-$J$ representation of $SO(d)\subset SO(d+1,1)$, i.e. a spinning conformal field, there are additional constraints afforded by the conformal symmetry group in the form of transversality with respect to a set of auxiliary vectors $z^{s_i}$ encoding the spin of the field, e.g. $z^{s_1}\Phi_{s_1\ldots s_J} =0$. We will take up this question in future work.} is that $\Phi(X)\equiv\Phi_\Delta(X)$ must be a homogeneous function of $X$ with scaling weight $\Delta$, i.e. $\Phi_\Delta(X)\to\Phi_\Delta(\lambda X) = \lambda^{-\Delta}\Phi_\Delta(X)$ under scale transformations $X\to \lambda X$.  Since our aim is to preserve a $p$-dimensional conformal subgroup, with $p<d$, we require that any conformal defect constructed from an NN must similarly be a homogeneous function of the subset of inputs spanning $\mathfrak{M}_p\hookrightarrow\mathbb{R}^d$ \cite{Billo:2016cpy}.  Further, if the additional normal bundle symmetry is preserved, then we also require that the network architecture transforms in an irreducible representation labeled by the `transverse spin' $s$.  

With these conditions in mind, we can begin to construct conformal defects from NNs. The first requirement in the parameter-space representation is that $P(\Theta)$ be chosen such that it is invariant under the defect conformal symmetry group.  This is clearly a necessary condition, but it is not sufficient to meet all of the criteria above. 

To illustrate this point, let us consider a $d$-dimensional conformal field $\Phi_\Delta(X)$ with scaling dimension $\Delta$ transforming as a scalar under $SO(d)$, and with a parameter distribution $P(\Theta)$ that is invariant under $SO(d+1,1)$, e.g. $\Phi_\Delta(X)= \Phi_\Delta(\Theta\cdot X)$ and $P(\Theta)= \CN(\Theta\mid 0,\sigma^2\mathds{1}_d)$. Now, let us fix the architecture and simply replace $P(\Theta)$ with a distribution $P^\prime(\Theta)$ which is invariant under $SO(p+1,1)\times SO(q)_N\subset SO(d+1,1)$.  In order to render the defect symmetries manifest in the architecture, we split $\Theta\cdot X \to \Theta\bullet X + \Theta\circ X$. The one-point function of $\Phi_\Delta(X)$ with parameters $\Theta\sim P^\prime$ is then
\begin{align}\label{eq:test-one-pt-fn}
    \mathbb{E}_{\Theta\sim P^\prime}\left[\Phi_\Delta(X)\right] = \int \mathrm{d}\Theta~ P^\prime(\Theta)~ \Phi_\Delta(\Theta\bullet X + \Theta\circ X)~.
\end{align}
If $\Theta\circ X\neq 0$, then this one-point function should take the form of \eq{defect-one-pt-fn-1}.  However, this functional form is not guaranteed, and unless the non-linearity is chosen carefully, it is not clear that the integral will (i)  be non-vanishing, (ii) have definite scaling dimension $\widehat{\Delta}$ under the defect conformal symmetry group, and (iii) result in a form that can be organized into irreducible representations of $SO(q)_N$. Indeed, even in the simple case of the toy models that we consider below, namely $\Phi_\Delta(X) = (\Theta\cdot X)^{-\Delta}$ with $\Delta <0$ and a factorized $P(\Theta) = \widehat{P}(\Theta)\widetilde{P}(\Theta)$ with $\widehat{P}$ ($\widetilde{P}$, resp.) invariant under $SO(p+1,1)$ ($SO(q)_N$, resp.), the integral above does not result in a simple monomial in $1/|X\circ X|^{\Delta}$. However, we will see below that there is a scheme in which to interpret correlation functions of `ambient fields' where the appearance of linear combinations of such terms is a feature, rather than a bug.

This insufficiency of a simple modification of the parameter distribution suggests that, in order to realize a conformal defect in an NN-FT, we must also allow for alterations of the architecture which fully capture the reduced symmetries.  Naively, we could factorize $\Phi(X)$ into the product of two auxiliary architectures
\begin{align}\label{eq:naive-defect-arch}
    \Phi(X) = \widehat{\varphi}(X) \widetilde{\varphi}(X)~,
\end{align}
where $\widehat{\varphi}(X)\equiv \widehat{\varphi}(\Theta\bullet X)$ and $\widetilde{\varphi}(X)\equiv\widetilde{\varphi}(\Theta\circ X)$ act on the tangent and normal spaces to the defect in the embedding space, respectively.  We would then require that the factorized architecture transform with definite scaling weight $\widehat{\Delta}$ under $SO(p+1,1)$ Weyl rescalings, i.e. $\widehat{\varphi}_{\Delta}(\lambda X) = \lambda^{-\widehat{\Delta}}\widehat{\varphi}_{\Delta}(X)$, and that $\widetilde{\varphi}(X)$ transform in an irreducible representation of $SO(q)_N$ labeled by its transverse spin $s$. The construction of an arbitrary defect is then specified by the choice of an appropriate parameter distribution $P^\prime(\Theta)$ and auxiliary architectures $\widehat{\varphi}$ and $\widetilde{\varphi}$.  

In generalizing \eq{naive-defect-arch}, we can take inspiration from the defect OPE, and instead decompose the `ambient field' into a sum over pairs $(\widehat{\varphi},~\widetilde{\varphi})$. Let us pause here for a moment to consider the implications of the above form of $\Phi_{\Delta}(X)$. Drawing from \eq{defect-ope}, we propose that, generically, the NN architecture
 \begin{align}\label{eq:general-defect-arch}
     \Phi_{\Delta}(X) = \sum_{\widehat{\Delta},s} \widehat{b}_{\widehat{\Delta},s}\widehat{\varphi}_{\widehat{\Delta}}(X)\widetilde{\varphi}_s(X)~
 \end{align}
engineers a conformal defect. Though the concept of expanding in the tower of `primaries' is straightforward and easy to understand from the NN perspective,\footnote{It should be stressed that what we are referring to as a `primary' in the NN sense does not easily map on to a `primary' in the CFT sense.  Here, we simply mean that a (scalar) `primary' NN is an architecture $\Phi_\Delta$ whose expectations have the form of correlation functions of (scalar) primary operators; in particular, the two-point function. We leave addressing the subtle points about conformal representation theory in NN-FTs to future work.} we leave the exploration of `spinning' blocks and descendants for future work.

For a certain class of theories, which we refer to as `monomial' in the following section, we observe that the sum over $(\widehat{\Delta},s)$ terminates at a finite order, dependent on the dimension $\Delta$ of the ambient field $\Phi_\Delta$.  The decomposition of $\Phi_\Delta$ into the form of \eq{general-defect-arch}, together with the modified parameter distribution $P^\prime(\Theta)$ suited to the defect conformal symmetries, will allow us to compute the one- and two-point correlation functions for both ambient $\Phi_\Delta(X)$ and defect $\widehat{\varphi}_{\widehat{\Delta}}(X)$ fields.  Importantly, for the $\langle \Phi_{\Delta_1}(X_1)\Phi_{\Delta_2}(X_2)\rangle$ correlator, the truncation of this OPE-like expansion in the defect channel will allow us to easily find closed-form expressions for the defect couplings and conformal blocks order-by-order in $\widehat{\Delta}$ and $s$, by iteratively solving the $SO(p+1,1)$ and $SO(q)_N$ Casimir equations. On the other hand, the same will not be true for the so-called `reciprocal' theories, for which, as expected, the defect OPE contains an infinite number of terms for any value of $\Delta>0$.

\section{Examples}\label{sec:examples}

In this section, we present concrete models that demonstrate the utility of the formalism presented above.  In particular, we will focus on a single family of models of the form $\Phi_\Delta(X) = (\Theta\cdot X)^{-\Delta}$. Throughout, we consider network parameters drawn from an isotropic Gaussian distribution $\Theta\sim P=\mathcal{N}(0,\sigma^2\mathds{1}_d)$, and we treat the cases with positive and negative scaling dimension separately. The latter class of toy models will be presented first, as it offers a direct path to computing expectations in the NN. Since these theories are characterized by $-\Delta=:\alpha >0$, we will make a useful, albeit abusive, distinction by referring to them as `monomial' theories, as opposed to the `reciprocal' theories making up the $\Delta >0$ class. We will explore the latter in a later subsection.

\subsection{Monomial NN-FTs}\label{sec:non-unitary}
In this section, we use the formalism discussed above to introduce conformal defects in a class of $\Delta<0$ theories studied in \cite{Halverson:2024axc}. We compute correlation functions involving both ambient and defect fields.  In doing so, we will show that the two-point correlation functions of ambient-ambient and ambient-defect fields are exactly solvable in terms of hypergeometric functions \cite{Billo:2016cpy}. This will allow us to easily compare the resulting expressions to the expectations derived from the defect OPE, and compute exact analytic expressions for the defect conformal blocks.  Lastly, motivated by the solution to the Casimir equations in the defect channel, we will propose a direct NN interpretation in terms of the spectrum of defect fields organized according to $SO(p+1,1)\times SO(q)_N$ quantum numbers.

The first example that we will consider in order to ground our discussion and illustrate the main points of our formalism is given by the following ambient network architecture 
\begin{align}
    \Phi_\alpha(X) = (X\cdot \Theta)^{\alpha}~,
\end{align}
where we defined $\alpha:=-\Delta>0$, i.e. the negative of the scaling weight under the action of dilatations, for notational simplicity.  For the remainder of this section, we will only consider the case $\alpha\in\mathbb{N}$, but for monomial theories, the analytic continuation to $\alpha\in\mathbb{R}_+$ is straightforward. By breaking the ambient global conformal symmetry group to its defect subgroup, we decompose into sub-architectures $\widehat{\varphi}$ and $\widetilde{\varphi}$, respectively encoding `defect' and `normal' networks as in \eq{defect-ope}. The preservation of defect conformal symmetry constrains the form of the sub-architectures to be homogeneous polynomials in $X\bullet\Theta$ and $X\circ \Theta$, respectively.  To make this clear, we can use \eq{dot-decomposition} and the binomial expansion of $\Phi_\alpha(X)$ to determine the structure of $\widehat\varphi$ and $\widetilde{\varphi}$ as follows,
\begin{align}
    (X\cdot\Theta)^\alpha = \sum_{\widehat{\alpha}=0}^\alpha \frac{\Gamma(\alpha+1)}{\Gamma(\widehat\alpha +1)\Gamma(\alpha-\widehat\alpha +1)}(X\bullet\Theta)^{\widehat{\alpha}}(X\circ\Theta)^{\alpha-\widehat{\alpha}}.
\end{align}
As described above, we also require an appropriate deformation of the distribution of network parameters.  For simplicity, we assume that $P(\Theta)$ factorizes into 
\begin{align}\label{eqdefect-distbn}
    P(\Theta) = \CN(\Theta\mid\vec{0}_{p+2},\widehat{\sigma}\mathds{1}_{p+2})~\CN(\Theta\mid\vec{0}_q,\widetilde{\sigma}\mathds{1}_q)=:\widehat{P}(\Theta)\widetilde{P}(\Theta)~,
\end{align}
i.e. the parameters in $X\bullet\Theta$ ($X\circ\Theta$, resp.) are distributed according to a multivariate
Gaussian distribution with zero mean and isotropic variance $\widehat{\sigma}$ ($\widetilde{\sigma}$, resp.).  With these choices of factorization of $P(\Theta)$ and decomposition of $\Phi_\alpha(X)$, we will compute various correlation functions involving both ambient and defect fields.

\subsubsection*{One-point functions}
To begin, let us compute, as a sanity check, the one-point function of the defect field $\widehat{\varphi}_{\alpha}(X) = (X\bullet \Theta)^{\widehat\alpha}$. We find
\begin{align}
    \langle\widehat\varphi(X)\rangle &= \mathbb{E}_{\widehat{P}(\Theta)}\left[\widehat\varphi(X)\right] \sim \widetilde\sigma^2 (X\bullet X)^{\widehat{\alpha}/2}. 
\end{align}
Since the field $\widehat\varphi(X)$ is restricted to the submanifold defined by $X\circ X=0$, and so $X\bullet X\overset{\rm P.S.}{=}0$, the pullback of $\langle\widehat\varphi(X)\rangle$ to the Poincar\'e section vanishes unless $\alpha=0$. This finding is in line with the expectation that the only $SO(p+1,1)$-scalar, defect-localized field with non-vanishing one-point function is the identity operator.  

Proceeding with the more interesting case of the one-point function of an ambient scalar field $\Phi_\alpha(X)$, we use the binomial expansion to find
\begin{align}\label{eq:monomial-one-pt-fn}\nonumber
    \langle\Phi_\alpha(X)\rangle&=\mathbb{E}_{P(\Theta)}\left[\Phi_\alpha(X)\right]\\\nonumber
    &= \sum_{\widehat{\alpha}=0}^\alpha \frac{\Gamma(\alpha+1)}{\Gamma(\widehat\alpha+1)\Gamma(\alpha-\widehat\alpha+1)}\mathbb{E}_{\widehat{P}(\Theta)}\left[(X\bullet \Theta)^{\widehat{\alpha}}\right]\mathbb{E}_{\widetilde{P}(\Theta)}\left[(X\circ \Theta)^{\alpha-\widehat{\alpha}}\right]\\
    &= \frac{\Gamma(2\alpha+1)}{2^\alpha\Gamma(\alpha+1)}(\widehat\sigma^2 X\bullet X+\widetilde\sigma^2 X\circ X)^\alpha \overset{\rm P.S.}{=}\frac{\Gamma(2\alpha+1)}{2^\alpha\Gamma(\alpha+1)}(\widehat\sigma^2-\widetilde\sigma^2)^\alpha \Vert X\Vert_\circ^{2\alpha}~.
\end{align}
  In moving from the second to the third line, we note that the expectation values $\mathbb{E}[\ldots]$ are only non-vanishing if $\widehat\alpha \in 2\mathbb{N}$ and $\alpha-\widehat\alpha\in2\mathbb{N}$. In the final equality, we restrict to the Poincar\'e section. In contrast to the defect-localized field $\widehat\varphi(X)$, the ambient field $\Phi_\alpha(X)$ is not restricted to the submanifold $X\circ X\overset{\rm P.S.}{=}0$, and we see that the one-point function of this defect architecture is precisely of the form that one would expect for the one-point function of an ambient scalar of dimension $\alpha = -\Delta$ in the presence of a conformal defect \cite{Billo:2016cpy}.  One further point to note is that the overall normalization vanishes when $\widehat{\sigma}=\widetilde{\sigma}$: in that case, the defect becomes trivial.

In the binomial expansion in the first line of \eq{monomial-one-pt-fn}, we see a structure similar to that of the defect OPE.  Indeed, looking back to \eq{defect-ope} and dropping the towers of descendants ($n=0$), we can draw a direct analogy with
\begin{align}
    \Phi_\alpha(X) = \sum_{\widehat{\alpha}}b_{\Phi \widehat{\varphi}} \widetilde{\varphi}_{\alpha-\widehat{\alpha}}(X) \widehat{\varphi}_{\widehat{\alpha}}(X)~.
\end{align}
Then, taking $\widehat{\varphi}_{\widehat{\alpha}}(X) = (X\bullet \Theta)^{\widehat\alpha}$ and $\widetilde{\varphi}_s = (X\circ \Theta)^s$, as well as truncating the defect spectrum to $\widehat\alpha\leq\alpha$, we recover the result above, with $b_{\Phi\widehat{\CO}} = \frac{\Gamma(\alpha+1)}{\Gamma(\widehat\alpha +1)\Gamma(\alpha-\widehat\alpha+1)}$. 

What is particularly interesting about the one-point function for this monomial theory is that the expansion in defect fields truncates at a finite order.  Typically, in defect CFTs, the one-point function of a scalar primary expanded into the defect OPE requires an infinite tower of defect primaries.  Looking back to the ambient theory on its own, the finite nature of this expansion may not be as surprising; indeed, the `non-unitary' $\alpha=1$ theory in \cite{Halverson:2024axc} had exactly solvable conformal blocks in the four-point function expressible as a sum of a few monomials in the cross ratios.  In the following subsection, where we consider the two-point function, which in the case of two ambient insertions is the analog of the ambient theory four-point function; we also find relatively simple, exact expressions, and again find that an expansion in terms of defect two-point functions terminates at a finite order.

\subsubsection*{Two-point functions}
We now compute two-point functions of the form $\mathbb{E}\left[\Phi_\alpha(X)^{k}\widehat\varphi(X)^{2-k}\right]$ with $k\in\{0,1,2\}$. That is, in the language of defect CFTs, we consider defect-defect, mixed defect-ambient, and ambient-ambient two-point functions. Due to the residual $SO(p+1,1)$ symmetry, the defect-defect (i.e. $k=0$) case is straightforward to read off:
\begin{align}\label{eq:defect-defect-2pt-fn}
    \mathbb{E}_{\widehat{P}(\Theta)}\left[\widehat\varphi_{\widehat{\alpha}_1}(X_1)\widehat\varphi_{\widehat{\alpha}_2}(X_2)\right] \overset{\mathrm{P.S.}}{=} \delta_{\widehat\alpha_1\widehat\alpha_2}\Gamma(\widehat\alpha_1+1)\widehat\sigma^{2\widehat\alpha_1}(X_1\bullet X_2)^{\widehat\alpha_1}~,
\end{align}
which is simply the two-point function one would expect from $SO(p+1,1)$-invariant NN conformal fields \cite{Halverson:2024axc}.

Less trivial is the mixed correlation function, given by $k=1$. Nevertheless, we can use the binomial expansion of $\widetilde{\varphi}_s(X)$ (or the defect OPE), as in the case of the one-point function, to find
\begin{align}\hspace{-1cm}
    \mathbb{E}\left[\Phi_{\alpha_1}(X_1)\widehat\varphi_{\widehat{\alpha}_2}(X_2)\right] &= \sum_{\widehat\alpha_1=0}^{\alpha_1}\frac{\Gamma(\alpha_1+1)}{\Gamma(\widehat{\alpha}_1+1)\Gamma(\alpha_1-\widehat\alpha_1+1)}\mathbb{E}_{\widehat{P}(\Theta)}\left[(X_1\bullet\Theta)^{\widehat{\alpha}_1}(X_2\bullet\Theta)^{\widehat{\alpha}_2}\right]\times\nonumber\\
    &\qquad\qquad \times\mathbb{E}_{\widetilde{P}(\Theta)}\left[(X_1\circ\Theta)^{\alpha_1-\widehat{\alpha}_1}\right]~.
\end{align}
We can then restrict to the Poincar\'e section and exploit the defect-defect expectation value in \eq{defect-defect-2pt-fn} to compute the expectation with respect to $\widehat{P}(\Theta)$ and collapse the sum to find a non-vanishing result only if $\widehat\alpha_2\leq \alpha_1$:
\begin{align}
    \mathbb{E}\left[\Phi_{\alpha_1}(X_1)\widehat\varphi_{\widehat{\alpha_2}}(X_2)\right] \overset{\rm P.S.}{=} \frac{\Gamma(\alpha_1+1)\Gamma(\widehat{\alpha}_2+1)}{\Gamma(\alpha_1-\widehat{\alpha}_2+1)}\widehat{\sigma}^{2\widehat\alpha_2}(X_1\bullet X_2)^{\widehat\alpha_2}\mathbb{E}_{\widetilde P(\Theta)}\left[(X_1\circ\Theta)^{\alpha_1-\widehat\alpha_2}\right]~.
\end{align}
The remaining expectation is non-vanishing only if $\alpha_1-\widehat{\alpha}_2 =: 2\alpha_{12}\in 2\mathbb{N}$, whence
\begin{align}
    \mathbb{E}\left[\Phi_\alpha(X_1)\widehat\varphi_{\widehat{\alpha}_2}(X_2)\right] \overset{\rm P.S.}{=} \frac{2^{\alpha_{12}}\widehat\sigma^{2\widehat{\alpha}_2}\Gamma(\alpha_1+1)}{\Gamma(\alpha_{12}+1)}(\widehat{\sigma}^2-\widetilde\sigma^2)^{\alpha_{12}}\Vert X_2\Vert_\circ^{2\alpha_{12}}(X_1\bullet X_2)^{\widehat{\alpha}_2}~.
\end{align}

The last case of two-point functions left to consider is the ambient-ambient (i.e. $k=2$) channel; compared to the cases above, this will require a somewhat greater amount of care, but it will also open up an interesting interpretative framework on the NN side, due to the emergence of defect conformal blocks. Although we will find an exact closed-form expression for this two-point function, it will be useful, in order to develop the interpretive framework in the NN picture, to also employ a description in terms of the defect OPE. We can use the binomial expansion as in the previous cases, and pull back to the Poincar\'e section to find 
\begin{align}
    \mathbb{E}[\Phi_{\alpha_1}(X_1)\Phi_{\alpha_2}(X_2)] = \Vert X_1\Vert^{\alpha_1}  \Vert X_2\Vert^{\alpha_2}  f(\chi,\psi)~,
\end{align}
where we have used the conformal cross ratios in \eq{cross-ratios}. Using the intermediate results in \app{two-point-fns}, the sums resulting from the binomial expansions have been expressed in
\begin{align}\label{eq:full-two-pt-fn-block}
    f(\chi,\psi) =\Gamma(\alpha_1+1)\Gamma(\alpha_2+1) \sum_{d_1=0}^{\alpha_1}\sum_{d_2=0}^{\alpha_2}\sum_{m_1,m_2}{}^{\!\!\!\!{{\text{\normalsize{$\prime$}}}}}(-1)^{\frac{d_1+d_2}{2}}\widetilde{{\Upsilon}}^{d_1,d_2}_{m_1}\widehat{\Upsilon}^{\alpha_1-d_1,\alpha_2-d_2}_{m_2}\left(\frac{\chi}{2}\right)^{m_1}\cos^{m_2}\psi,
\end{align}
where we defined
\begin{align}
    \Upsilon^{a,b}_c \equiv \frac{2^{\frac{2c-a-b}{2}}\sigma^{a+b}}{\Gamma(c+1)\Gamma(\frac{a-c+2}{2})\Gamma(\frac{b-c+2}{2})}
\end{align}
with $\widetilde{\Upsilon}$ ($\widehat{\Upsilon}$, resp.) indicating $\widetilde\sigma$ ($\widehat\sigma$, resp.) and the restricted sum
\begin{align}
    \sum_{m_1,m_2}{}^{\!\!\!\!{{\text{\normalsize{$\prime$}}}}} \equiv \sum_{\substack{
        m_1=0\\
        d_1-m_1\in2\mathbb{N}\\
        d_2-m_1\in2\mathbb{N}}}^{\operatorname{min}(d_1,d_2)}\sum_{\substack{
        m_2=0\\
        \alpha_1-d_1-m_2\in2\mathbb{N}\\
        \alpha_2-d_2-m_2\in2\mathbb{N}}}^{\operatorname{min}(\alpha_1-d_1,\alpha_2-d_2)}.
\end{align}
We can then resum $f(\chi,\psi)$ using hypergeometric functions (w.l.o.g., choosing $\alpha_1> \alpha_2$)
\begin{align}\label{eq:general-two-pt-fn}
    &\mathbb{E}[\Phi_{\alpha_1}(X_1)\Phi_{\alpha_2}(X_2)] = \frac{\Gamma(\alpha_2+1)}{2^{\frac{\alpha_2+\alpha_1}{2}}\Gamma(\frac{\alpha_1-\alpha_2+2}{2})}(\widetilde\sigma^2-\widehat{\sigma}^2)^{\frac{\alpha_1-\alpha_2}{2}} \Vert X_1\Vert_\circ^{\alpha_1}  \Vert X_2\Vert_\circ^{\alpha_2} \times\\\nonumber
    &\times \left(2\widetilde{\sigma}^2\cos\psi - \widehat{\sigma}^2\chi\right)^{\alpha_2}{}_2F_1\left(\frac{1-\alpha_2}{2},-\frac{\alpha_2}{2},\frac{\alpha_1-\alpha_2+2}{2},\frac{4(\widetilde\sigma^2-\widehat\sigma^2)^2}{(2\widetilde\sigma^2\cos\psi-\widehat\sigma^2\chi)^2}\right),
\end{align}
where the right-hand side is non-vanishing only if $\alpha_1+\alpha_2\in 2\mathbb{N}$, and where ${}_2F_1$ denotes the Gaussian hypergeometric function. Note that choosing $\alpha_2>\alpha_1$ instead amounts to exchanging the indices $1\leftrightarrow 2$.

As an aside, we note that it is possible to bypass the binomial expansion and instead derive \eq{general-two-pt-fn} in a different fashion. Specifically, if we denote the covariance matrix for the joint multivariate normal distribution by $\Sigma = \operatorname{diag}(\widehat\sigma^2\mathds{1}_{p+1},\widetilde\sigma^2\mathds{1}_q)$, we can define the moment generating function (MGF) for the two-point correlation function as
\begin{align}\label{eq:gaussian-mgf}
    M(s,t) = \mathbb{E}\left[e^{s X_1\cdot\Theta +t X_2\cdot\Theta}\right] = e^{\frac{1}{2}(s^2 X_1^{\rm T}\Sigma X_1 + t^2 X_2^{\rm T}\Sigma X_2) + s t X_1^{\rm T}\Sigma X_2}~,
\end{align}
where $X_i^{\rm T}\Sigma X_j = \widehat{\sigma}^2X_i\bullet X_j + \widetilde{\sigma}^2 X_i\circ X_j$. Then, functionally differentiating with respect to the `sources' $s$ and $t$,
\begin{align}
    \mathbb{E}\left[\Phi_{\alpha_1}(X_1)\Phi_{\alpha_2}(X_2)\right] =\left. \frac{\partial^{\alpha_1}}{\partial s^{\alpha_1}}\frac{\partial^{\alpha_2}}{\partial t^{\alpha_2}}M(s,t)\right|_{s=t=0}~,
\end{align}
using the trinomial expansion to extract the $s^{\alpha_1}t^{\alpha_2}$ term, and resumming the result, we find exactly the same expression in \eq{general-two-pt-fn}.  This method of computing the correlation functions will be particularly useful in analytically continuing to positive scaling dimensions, $\Delta = -\alpha >0$, in the following subsection.

Having computed the exact form for the ambient-ambient two-point function, we can now proceed to compare with what we would find by employing the `defect OPE'.  In standard defect CFTs, the two-point function $\langle\CO_1\CO_2\rangle$ of ambient scalar primaries can be computed by taking the OPE in two different channels related by crossing symmetry, namely the defect and ambient channels \cite{Billo:2016cpy}
\begin{align}
    \begin{split}
        \langle\CO_1(X_1)\CO_2(X_2)\rangle =&\,\, \Vert X_1 \Vert_{\circ}^{\alpha_1}\Vert X_2 \Vert_{\circ}^{\alpha_2}\sum_{\widehat{\CO},s} \widehat{b}_{1\widehat\CO}\widehat{b}_{2\widehat\CO}\widehat{f}_{\widehat\alpha,s}(\chi,\psi)~, \hspace{1.6cm}{\rm(defect)}\\
         =& \,\,\xi^{\frac{\alpha_1+\alpha_2}{2}}\Vert X_1 \Vert_{\circ}^{\alpha_1}\Vert X_2 \Vert_{\circ}^{\alpha_2}\sum_{k,J}C_{12k}a_k f_{\alpha_k,J}(\xi,\psi)~, \quad{\rm(ambient)}
    \end{split}
\end{align}
where in the ambient channel it is more convenient to employ the cross ratio
\begin{align}
    \xi = - \frac{2X_1\cdot X_2}{ \Vert X_1\Vert_\circ~ \Vert X_2\Vert_\circ}
\end{align}
than the aforementioned cross ratio $\chi$. The block appearing in each summand must satisfy the appropriate Casimir equation: $SO(p+1,1)\times SO(q)_N$ for the defect channel and $SO(D+1,1)$ for the ambient channel.  Because the defect symmetry factorizes into a product group, the Casimir equations for the tangent and normal bundle symmetries can be solved independently \cite{Billo:2016cpy}.

Focusing firstly on the defect channel, the ambient Casimir operator $J_{MN}J^{MN} = X_{[M}\partial_{N]}X^{[M}\partial^{N]}$ decomposes into the $SO(p+1,1)$ Casimir $\widehat{J}_{AB}\widehat{J}^{AB}$ and the $SO(q)_N$ Casimir $\widetilde{J}_{IJ}\widetilde{J}^{IJ}$.  Defining $C_{\widehat\alpha,s}\equiv \widehat\alpha(\widehat\alpha-p)+s(s+q-2)$, we need to simultaneously solve the eigenvalue equations
\begin{subequations}
\begin{align}\label{eq:sop-casimir-eqn}
        (\widehat{J}_{AB}\widehat{J}^{AB}+2C_{\alpha,0})\widehat{f}_{\widehat\alpha,s} &=0\\\label{eq:soq-casimir-eqn}
        (\widetilde{J}_{IJ}\widetilde{J}^{IJ}+2C_{0,s})\widehat{f}_{\widehat\alpha,s}&=0~.
\end{align}
\end{subequations}
Before proceeding with the more general case $\alpha_1\neq\alpha_2$, let us consider a simple example that will illustrate an interpretation of these defect blocks. In particular, consider the limiting case $\alpha_1=\alpha_2=\alpha$. When $\alpha = 1$,
\begin{align}
f(\chi,\psi) = \widetilde{\sigma}^2\cos\psi - \frac{\widehat{\sigma}^2}{2}\chi~.  
\end{align}
With some algebraic manipulations, we can easily rearrange this in to defect blocks that satisfy the appropriate Casimir equations with specified $SO(p+1,1)$ scaling weight $\widehat{\alpha}$ and $SO(q)_N$ spin $s$.  Since $\alpha_1=\alpha_2$, we denote $\widehat{b}_{1\widehat{\CO}}\widehat{b}_{2\widehat{\CO}}=: \widehat{b}_{\widehat{\alpha},s}$ and find
\begin{align}
\begin{split}
 f(\chi,\psi) &= \widehat{b}_{0,0}\widehat{f}_{0,0} +\widehat{b}_{1,0}\widehat{f}_{1,0} + \widehat{b}_{0,2}\widehat{f}_{0,2}\\
&=\frac{\widetilde{\sigma}^2}{q} + \left(-\frac{\widehat{\sigma}^2}{2}\right)\chi + \frac{\widetilde{\sigma}^2}{q}\left( q\cos\psi -1\right)~.
\end{split}
\end{align}
In the case $\alpha=2$,
\begin{align}
    f(\chi,\psi) = 2\widehat{\sigma}^4\cos^2\psi + \frac{\widetilde{\sigma}^4}{2}\chi^2 - 2\widehat{\sigma}^2\widetilde{\sigma}^2\chi\cos\psi +(\widehat{\sigma}^2-\widetilde{\sigma}^2)^2~.
\end{align}
Choosing the $s>0$ and $\alpha >0$ branches, we can write
\begin{align}
    f(\chi,\psi) = \widehat{b}_{0,0}\widehat{f}_{0,0} + \widehat{b}_{1,0}\widehat{f}_{1,0}+ \widehat{b}_{0,2}\widehat{f}_{0,2} + \widehat{b}_{1,2}\widehat{f}_{1,2} + \widehat{b}_{2,0}\widehat{f}_{2,0} +\widehat{b}_{0,4}\widehat{f}_{0,4}~,
\end{align}
where the defect blocks $\widehat{f}_{\widehat{\alpha},s}$ and couplings $\widehat{b}_{\widehat{\alpha},s}$ are given by
\begin{align}
\begin{split}
    \widehat{b}_{0,0}\widehat{f}_{0,0} &= (\widehat\sigma^2-\widetilde{\sigma}^2)^2+\frac{2\widetilde{\sigma}^4}{p+2}+\frac{6\widehat{\sigma}^4}{q(4+q)}~,\\
    \widehat{b}_{1,0}\widehat{f}_{1,0} &= -\frac{2\widehat{\sigma}^2\widetilde{\sigma}^2}{q}\chi~,\\
    \widehat{b}_{0,2}\widehat{f}_{0,2} &= \frac{12\widehat{\sigma}^4}{q(q+4)}(q\cos\psi -1)~,\\
   \widehat{b}_{1,2}\widehat{f}_{1,2} &= -\frac{2\widehat{\sigma}^2\widetilde{\sigma}^2}{q}\chi(q\cos\psi-1)~,\\
   \widehat{b}_{2,0}\widehat{f}_{2,0} &=\frac{\widetilde{\sigma}^4}{2}\Big(\chi^2 - \frac{4}{2+p}\Big)~,\\
   \widehat{b}_{0,4}\widehat{f}_{0,4}&= 2\widehat{\sigma}^4\Big(\cos^2\psi -\frac{6\cos\psi}{q+4}+\frac{3}{(q+2)(q+4)}\Big)~.
 \end{split}
\end{align}
Generalizing to arbitrary $\alpha_1=\alpha_2$, a block with arbitrary $SO(p+1,1)$ and $SO(q)_N$ quantum numbers ($\widehat\alpha,s$) takes the form
\begin{align}\label{eq:general-casimir-soln}
    \widehat{f}_{\widehat{\alpha},s} = \chi^{\widehat{\alpha}}\cos^s\psi~ {}_2F_1\left(\frac{1-\widehat{\alpha}}{2},-\frac{\widehat\alpha}{2},\frac{2-2\widehat{\alpha}-p}{2},\frac{4}{\chi^2}\right){}_2F_1\left(\frac{1-2s}{2},-s, \frac{4-4s-q}{2}, \sec\psi\right)
\end{align}
as expected from symmetry arguments \cite{Billo:2016cpy}.

In these simple examples, we can find the basis for interpreting the defect conformal fields as a neural network.  The fundamental lesson that we would like to emphasize here is that explicitly breaking conformal symmetries realized either in the space of inputs or parameters gives rise to a space of `fundamental' networks, here represented by the defect fields in the various `exchange channels' labeled by their $(\widehat{\alpha},s)$ quantum numbers.   What the simple ambient-ambient two-point function above demonstrates is that, unlike the expansion of the one-point function which could be fully reconstructed from the $s=0$ primaries, defect fields with non-zero spin under $SO(q)_N$ are crucial components in the expansion of the ambient correlator in the defect spectrum.  Operationally, on the NN side, this means that in order to compute expectation values for any $SO(d+1,1)$-invariant network, e.g. $n$-point correlation functions, with respect to a parameter distribution that only respects an $SO(p+1,1)\times SO(q)_N$ symmetry, one can instead consider a weighted sum of expectations of explicitly $SO(p+1,1)$-invariant, homogeneous networks with prescribed scaling-weight parameters $\widehat\alpha$ transforming in specific irreducible representations of $SO(q)_N$.  

Solving the ambient channel, though guaranteed to be equivalent to the defect expansion by crossing symmetry, is not generally possible, and often one needs to limit oneself to a lightcone expansion in order to find a closed-form expression \cite{Billo:2016cpy}.  Moreover, this would require a careful analysis of the ambient `OPE-like' expansion, which would in turn necessitate a study of spinning NN conformal fields, i.e. fields with non-trivial $SO(d+1,1)$ spin, as opposed to defect fields with transverse spin under the global $SO(q)_N$ symmetry. We leave the analysis of such spinning NN conformal fields for future work.

\subsection{Reciprocal NN-FTs}\label{sec:unitary}

In this subsection, we construct defects from NNs with positive scaling dimensions. We consider fields of the form $\Phi_\Delta(X) = (\Theta\cdot X)^{-\Delta}$, where the parameters $\Theta$ are drawn from a Gaussian distribution. We will refer to such models as `reciprocal', in contrast to the monomial theories discussed in the previous sections. Note that, despite having positive dimension, these theories are not necessarily unitary. For instance, a `free scalar' NN field theory would need to satisfy the usual bound $\Delta \geq \frac{d-2}{2}$  in order to be properly classified as `unitary'; here, we only restrict to dimensions satisfying the weaker bound $\Delta>0$.

From the perspective of computing expectations over the network parameters, such reciprocal models are more challenging due to the singularity at the origin of parameter space. Nevertheless, with standard techniques from the study of QFTs, we can easily evaluate these expectations and compute correlators of the field $\Phi$ in the presence of a defect.

\subsubsection*{Defining the integral}
As in the monomial case, we wish to consider correlators built from the product of two Gaussian distributions, $P(\Theta)=\mathcal{N}(\Theta\mid0,\hat\sigma^2\mathds{1}_{p+2})~\mathcal{N}(\Theta\mid0,\tilde\sigma^2\mathds{1}_q)$. Crucially, however, the positive nature of the scaling dimension $\Delta$ of the NN architecture renders the naive integrals ill-defined. We can circumvent this problem by assigning a particular value to these integrals, through a combination of analytic continuation and regularization. Let us describe this procedure below.\bigskip

Given an ambient conformal field $\Phi$ with (positive) scaling dimension $\Delta$, we may consider its formal expansion as a series in the normal distance to the defect, 
\begin{align}\label{eq:reciprocal-defectOPE}
    \Phi_\Delta(X) &= (X\bullet\Theta+X\circ\Theta)^{-\Delta} = \sum_{n=0}^{\infty}\frac{(\Delta)_n}{n!}(-1)^n\widetilde{\varphi}_{-n}(X)\widehat{\varphi}_{\Delta+n}(X)~,
\end{align}
where $\widehat\varphi_{\Delta+n}(X) = (X\bullet\Theta)^{-\Delta-n}$ is a defect network and $\widetilde{\varphi}_{-n}(X)=(X\circ\Theta)^{n}$ is a normal network, as in \eq{defect-ope}. This expansion displays two important structures. Firstly, for any $\Delta>0$, the only normal fields $\widetilde{\varphi}$ that contribute to the ambient field $\Phi$ are those with negative scaling dimension; hence, the reciprocal theory will mimic a monomial theory in that subsector, with well-defined correlators given in terms of the moments of the Gaussian distribution $\widetilde{P}({\Theta})=\mathcal{N}(\Theta\mid0,\tilde\sigma^2\mathds{1}_q)$. The other important structure that appears in \eq{reciprocal-defectOPE} is a tower of defect fields $\widehat{\varphi}$ with strictly increasing (positive) scaling dimensions. This tower of fields starts at $\Delta>0$; therefore, it never crosses the identity defect operator, which explains the trivial result for the one-point function of an ambient field insertion detailed below. 

Since correlators within the normal subsector are well-defined Gaussian integrals, we focus our attention on the defect subsector, where one must assign values to the singular integrals of $\widehat{\varphi}$, whose evaluation yields negative moments of the distribution. Indeed, a generic $n$-point correlator of defect fields takes the following form,
\begin{align}
    G^{(n)}(\Delta_1,\ldots,\Delta_n) = \int\mathrm{d}^{p+2}\Theta~\frac{\widehat{P}(\Theta)}{(X_1\bullet\Theta)^{\Delta_1}\cdots(X_n\bullet\Theta)^{\Delta_n}}~.
\end{align}
Due to the singularities encountered at $\Theta\bullet X_i=0$, the integral above is ill-defined. Using the standard Feynman parametrization 
\begin{align}\label{eq:FeynmannParam}
    \frac{1}{(X_1\bullet\Theta)^{\Delta_1}\cdots(X_n\bullet\Theta)^{\Delta_n}} &= \frac{\Gamma(\Delta_1+\cdots+\Delta_n)}{\Gamma(\Delta_1)\cdots\Gamma(\Delta_n)}\int_{[0,1]^n}\mathrm{d}^nu~\frac{\delta(1-\sum_{k=1}^nu_k)\prod_{k=1}^nu^{\Delta_k-1}_k}{(\sum_{k=1}^nu_k( X_k\bullet\Theta))^{\sum_{k=1}^n\Delta_k}}~,
\end{align}
we can shift the divergence to the boundary of the $u$-integration region. In principle, further care must be taken upon treating \eq{FeynmannParam}, as it is only valid when all monomials $\Theta\bullet X_k$ have the same sign. Nevertheless, as we will discuss shortly, we can adopt a regularization scheme which, by suitably removing all divergences, will allow us to forgo this condition.\bigskip

We define the $n$-point correlator to be the following regulated integral, for $\Delta_k>0$,
\begin{align}
    G^{(n)}(\Delta_1,\ldots,\Delta_n) := \frac{\Gamma(\sum_{i=1}^n\Delta_i)}{\prod_{i=1}^n\Gamma(\Delta_i)}\sqint\, \mathrm{d}^nu\int\mathrm{d}^{p+2}\Theta~\widehat{P}(\Theta)\frac{\delta(1-\sum_{k=1}^nu_k)\prod_{k=1}^nu^{\Delta_k-1}_k}{(\sum_{k=1}^nu_k( X_k\bullet\Theta))^{\sum_{k=1}^n\Delta_k}}~,
\end{align}
 where we employ a particular regularization of the $u$-integral, denoted $\sqint\mathrm{d}^nu$, and an analytic continuation of the Gaussian integral $\int\mathrm{d}^{p+2}\Theta~\widehat{P}(\Theta)$. Whenever one of the scaling dimensions vanishes, we expect the $n$-point correlator to reduce to an $(n-1)$-point correlator. We will find this to indeed be the case. We also define the initial value $G^{(0)}=1$.

The analytic continuation of the Gaussian integral can be directly computed as follows. Denote the $u$-linear combination of vectors by $Z$, where $Z=\sum_{k=1}^nu_kX_k$. Without loss of generality, we can perform an $SO(p+2)$ rotation on the network parameters to align the $\Theta$-basis vector with $Z$, leading to a simplification of the $\Theta$-integral down to a one-dimensional Gaussian integral
\begin{align}
    \int\mathrm{d}^{p+2}\Theta~\frac{\widehat{P}(\Theta)}{(\sum_{k=1}^nu_k( X_k\bullet\Theta))^{\sum_{k=1}^n\Delta_k}} &= (Z\bullet Z)^{-\frac{1}{2}\sum_{k=1}^n\Delta_k}\int\mathrm{d}\Theta_Z~\widehat{P}(\Theta_Z)\Theta_Z^{-\sum_{k=1}^n\Delta_k}~.
\end{align}
The principal value of this integral vanishes when $\Delta=\sum_{k=1}^n\Delta_k\in 2\mathbb{Z}+1$; for every other value of $\Delta$, we use the analytic continuation from $\Delta<0$ to write
\begin{align}\label{eq:reciprocal-GaussianAnalyticCont}
    \int\mathrm{d}\Theta_Z~\widehat{P}(\Theta_Z)\widehat{\Theta}_Z^{-\Delta} &= \frac{1}{\sqrt{\pi}}(2\hat\sigma^2)^{-\frac{\Delta}{2}}\Gamma\left(\frac{1-\Delta}{2}\right).
\end{align}
Performing the $\Theta$-integration as described above leaves us with the following $u$-integral for the $n$-point correlator of defect fields $\widehat{\varphi}$,
\begin{align}\label{eq:reciprocal-integralDef}
    G^{(n)}(\Delta_1,\ldots,\Delta_n) &= \frac{1}{\sqrt{\pi}}\Gamma\left(\frac{1-\Delta}{2}\right)\frac{\Gamma(\Delta_1+\cdots+\Delta_n)}{\Gamma(\Delta_1)\cdots\Gamma(\Delta_n)}\times\\
    &\quad\times\sqint\, \mathrm d^nu~\delta\left(1-\sum_{k=1}^nu_k\right)\prod_{k=1}^nu^{\Delta_k-1}_k(2\hat\sigma^2Z(u)\bullet Z(u))^{-\frac{\Delta}{2}},\nonumber
\end{align}
where, as a reminder, we defined
\begin{align}
    Z(u) &= \sum_{k=1}^nu_kX_k & & \text{and}& \Delta=&\sum_{k=1}^n\Delta_k~.
\end{align}

The $u$-integration is constructed as follows. We start by truncating the integration domain for each $u$-variable at each end of the unit interval; with $0<\epsilon<\frac{1}{2}$ describing this cutoff, the domain becomes $n$ copies of the interval $[\epsilon,1-\epsilon]$. Then, consider the asymptotic series as $\epsilon$ tends to zero, $\epsilon\rightarrow 0$. Finally, isolate the coefficient in front of the log-term as the result. Symbolically, this amounts to the following definition
\begin{align}
    \sqint\mathrm d^nu &= \left.\int_{[\epsilon,1-\epsilon]^n}\mathrm d^nu\right|_{\ln(1/\epsilon)},
\end{align}
where $\cdots|_{\ln(\epsilon)}$ instructs one to keep the coefficient in front of the log-term only. This selection rule applied to the hard-cutoff scheme ensures that any functional redefinition of $\epsilon$, $\epsilon\mapsto\epsilon(\nu)$, will preserve the result. It also turns out that any smooth, $u\leftrightarrow 1-u$ symmetric, and sufficiently rapidly vanishing regulating function $\psi_\epsilon(u)$ also preserves the log-term coefficient, rendering our definition regulator-independent.\bigskip

Let us illustrate this cut-off regularization explicitly in the case of two $u$-integrals, corresponding to correlators between two defect insertions. For now, let $I(\Delta_1,\Delta_2)$ denote the $u$-integral part of $G^{(2)}(\Delta_1,\Delta_2)$,
\begin{align}
    I(\Delta_1,\Delta_2) &= \sqint\mathrm d^2 u~\delta\left(1-u_1-u_2\right)u_1^{\Delta_1-1}u_2^{\Delta_2-1}(Z(u)\bullet Z(u))^{-\frac{\Delta_1+\Delta_2}{2}}\nonumber\\
    &= \sqint\mathrm du~u^{\Delta_1-1}(1-u)^{\Delta_2-1}\times\\
    &\quad \times\left(u^2 X_1\bullet X_1 + 2u(1-u)X_1\bullet X_2 + (1-u)^2X_2\bullet X_2\right)^{-\frac{\Delta_1+\Delta_2}{2}}.\nonumber
\end{align}
As these are correlators between defect fields, their insertions are restricted to lie on the defect Poincar\'e section. This condition sets two of the terms to zero, and the $u$-integral becomes a regulated beta-function integral
\begin{align}
    I(\Delta_1,\Delta_2) &= (2X_1\bullet X_2)^{-\frac{\Delta_1+\Delta_2}{2}}\sqint\mathrm du~u^{\frac{\Delta_1-\Delta_2}{2}-1}(1-u)^{\frac{\Delta_2-\Delta_1}{2}-1}.
\end{align}
The resulting asymptotic expansion, as $\epsilon\rightarrow 0$, will generically depend on the values of $\Delta_{1,2}$. However, the log-term will be non-zero iff $\Delta_1=\Delta_2$, which yields the result that will later be used to determine the defect-defect two-point correlator,
\begin{align}
    I(\Delta_1,\Delta_2) &= 2(2X_1\bullet X_2)^{-\Delta_1}\delta_{\Delta_1,\Delta_2}~.
\end{align}

\subsubsection*{One-point functions}
Having properly defined the integration of defect fields $\widehat{\varphi}$ in \eq{reciprocal-integralDef}, we can evaluate the various correlators of ambient and defect fields, $\Phi$ and $\widehat{\varphi}$. The one-point correlation functions of defect fields $\widehat{\varphi}$ directly follow from the definition of $G^{(1)}$. Furthermore, one finds that for all $\Delta>0$, these vanish, as expected for any defect conformal primary (c.f. \eq{defect-one-pt-fn-1}). The expansion of the ambient field in terms of normal and defect ones in \eq{reciprocal-defectOPE} informs us that the ambient field one-point function also vanishes,
\begin{align}
    \mathbb{E}[\widehat\varphi_{\widehat{\Delta}}(X)]&=0 & &\text{and}& \mathbb{E}[\Phi_\Delta(X)]&=0~.
\end{align}
This may appear to be an amusing configuration, in that the coupling to the defect identity vanishes, but it reflects the observation that the expansion never involves the identity defect field.

\subsubsection*{Two-point functions}
Consider now the three classes of two-point correlation functions: $\mathbb{E}[\widehat{\varphi}_{\widehat{\Delta}_1}(X_1)\widehat{\varphi}_{\widehat{\Delta}_2}(X_2)]$, $\mathbb{E}[\Phi_\Delta(X_1)\widehat{\varphi}_{\widehat{\Delta}}(X_2)]$, and $\mathbb{E}[\Phi_{\Delta_1}(X_1)\Phi_{\Delta_2}(X_2)]$. The two-point function between defect fields $\widehat{\varphi}_{\widehat{\Delta}}$ is, by definition, the integral $G^{(2)}$ defined in \eq{reciprocal-integralDef}. A detailed calculation of the $u$-integral was given in the previous subsection. The full correlator for $\Delta\in\mathbb{R}_+$ reads
\begin{align}\label{eq:reciprocal-DD}
    \mathbb{E}[\widehat{\varphi}_{\widehat{\Delta}_1}(X_1)\widehat{\varphi}_{\widehat{\Delta}_2}(X_2)] &= \frac{\sec(\pi\widehat{\Delta}_1)}{\Gamma({\widehat{\Delta}}_1)}(\hat\sigma^2 X_1\bullet X_2)^{-{\widehat{\Delta}}_1}\delta_{{\widehat{\Delta}}_1,{\widehat{\Delta}}_2}~,
\end{align}
where $\sec(\pi {\widehat{\Delta}}) = (-1)^{\widehat{\Delta}}$ for ${\widehat{\Delta}}\in \mathbb{N}$. Note that the prefactor is meromorphic in ${\widehat{\Delta}}_1$ with poles at ${\widehat{\Delta}}_1 = \frac{2n+1}{2}$ for $n\in\mathbb{N}$.

Up to a subtle point that we will address below, \eq{reciprocal-DD} is precisely of the expected form for a conformal defect two-point function.  It also recovers the vanishing of the one-point functions, by taking the limit that either ${\widehat{\Delta}}_{1,2}\to 0$.  The subtlety that emerges is that the defect two-point function does not appear to respect reflection positivity for all values of ${\widehat{\Delta}}$.  Indeed, on the interval $\sec(\pi{\widehat{\Delta}})$ flips sign as it crosses poles at odd half-integral values.  Hence, if the nearest integer to $\Delta$ is odd, $\sec(\pi{\widehat{\Delta}})<0$ and reflection positivity is not obeyed.  We will expand on the generic failure of positivity in the two-point function for certain values of ${\widehat{\Delta}}$ below.

Before addressing the issue concerning reflection positivity, let us evaluate the mixed ambient-defect two-point function, which requires two separate steps. First, one expands the ambient field as in \eq{reciprocal-defectOPE} in terms of defect fields $\widehat{\varphi}$. The mixed correlator becomes an infinite sum of ambient one-point and defect two-point functions. For the former we use \eq{monomial-one-pt-fn} (or equivalently \eq{reciprocal-GaussianAnalyticCont}), while the latter is given in \eq{reciprocal-DD}. The resulting sum collapses down to
\begin{align}\label{eq:reciprocal-AD}
    \mathbb{E}[\Phi_\Delta(X_1)\widehat{\varphi}_{\widehat{\Delta}_2}(X_2)] &= \frac{2^{\frac{\Delta_1-{\widehat{\Delta}_2}}{2}}\sec(\pi{\widehat{\Delta}_2})}{\Gamma(\Delta_1)\Gamma(\frac{{\widehat{\Delta}_2}-\Delta_1}{2}+1)}\tilde\sigma^{{\widehat{\Delta}_2}-\Delta_1}\Vert X_1\Vert_\circ ^{{\widehat{\Delta}_2}-\Delta_1}(\hat\sigma^2 X_1\bullet X_2)^{-{\widehat{\Delta}_2}}~,
\end{align}
when ${\widehat{\Delta}_2}\geq\Delta_1$ and vanishes otherwise. Again, this has the structural form expected for a mixed ambient-defect correlator respecting defect conformal symmetry.  Also, we note that taking the ambient field to be the identity, and therefore $\Delta_1\to0$, recovers the vanishing of the one-point function for defect fields, and trivially also recovers the vanishing of the one-point function for ambient fields for any $\Delta_1>0$ when taking ${\widehat{\Delta}_2}\to 0$.

Finally, the correlator between two ambient fields mirrors the construction of the ambient-defect two-point function. One starts by expanding both ambient fields in terms of normal and defect ones. The correlator becomes an infinite sum of normal and defect field two-point functions. For the former we use lemma~\ref{lemma:2ptGaussianDefect}, while the latter is again given in \eq{reciprocal-DD}. The resulting sum can then be recast into a hypergeometric form,
\begin{align}\label{eq:reciprocal-AA}
    \mathbb{E}[\Phi_{\Delta_1}(X_1)\Phi_{\Delta_2}(X_2)]&= \frac{2^{\frac{\Delta_2-\Delta_1}{2}}\sec(\pi\Delta_1)}{\Gamma(\frac{\Delta_1-\Delta_2}{2}+1)\Gamma(\Delta_2)}\frac{(\hat\sigma^2 X_1\bullet X_2+\tilde\sigma^2 X_1\circ X_2)^{-\Delta_1}}{\tilde\sigma^{\Delta_2-\Delta_1}\Vert X_2\Vert_\circ ^{\Delta_2-\Delta_1}}\times\\\nonumber
    &\qquad\times{}_2F_1\left(\frac{\Delta_1}{2},\frac{\Delta_1+1}{2};\frac{\Delta_1-\Delta_2}{2}+1;\frac{\tilde\sigma^4\Vert X_1\Vert_\circ^2 \Vert X_2\Vert_\circ^2}{\left(\tilde\sigma^2X_1\circ X_2+\hat\sigma^2 X_1\bullet X_2\right)^{2}}\right),
\end{align}
for $\Delta_1+\Delta_2$ even and $\Delta_1\geq \Delta_2$. The $\Delta_1\leq \Delta_2$ case can be found by flipping $\Delta_1\leftrightarrow\Delta_2$ and $X_1\leftrightarrow X_2$. When $\Delta_1+\Delta_2$ is odd, the above expression vanishes.  

Before concluding, we note that we can equivalently express the two point function in terms of conformal defect cross ratios:
\begin{align}\label{eq:reciprocal-ambient-2pt-fn-2}
    \mathbb{E}[\Phi_{\Delta_1}(X_1)\Phi_{\Delta_2}(X_2)] &= \frac{2^{\frac{\Delta_2-\Delta_1}{2}}\sec(\pi\Delta_1)\tilde\sigma^{-\Delta_1-\Delta_2}}{\Gamma(\frac{\Delta_1-\Delta_2}{2}+1)\Gamma(\Delta_2)}\frac{(-\frac{1}{2}\frac{\hat\sigma^2}{\tilde\sigma^2}\chi+\cos(\psi))^{-\Delta_1}}{\Vert X_1\Vert_\circ^{\Delta_1}\Vert X_2\Vert_\circ^{\Delta_2}}\times\\\nonumber
    &\qquad\times{}_2F_1\left(\frac{\Delta_1}{2},\frac{\Delta_1+1}{2};\frac{\Delta_1-\Delta_2}{2}+1;\left(-\frac{1}{2}\frac{\hat\sigma^2}{\tilde\sigma^2}\chi+\cos(\psi)\right)^{-2}\right)~.
\end{align}
Moreover, the functional dependence of $\mathbb{E}[\Phi_{\Delta_1}(X_1)\Phi_{\Delta_2}(X_2)]$ on the cross ratios is essentially the same as in \eq{general-two-pt-fn}, with the replacement $\alpha_i\to-\Delta_i$. Hence, the analysis of the conformal blocks proceeds in exactly the same way as in the monomial case. The caveat here is that, unlike in the monomial class of theories, the defect OPE does not truncate at finite order for fixed $\Delta_i$.  This can be clearly seen from taking $\Delta_1=\Delta_2=1$ whence $\mathbb{E}[\Phi_{\Delta_1}(X_1)\Phi_{\Delta_2}(X_2)] \sim 1/\sqrt{(2\tilde\sigma^2\cos\psi-\hat\sigma^2 \chi)^2-4\tilde\sigma^4}$, and so leads to infinite power series expansions in small $\chi$ or $\cos\psi$.  Nevertheless, we can resum the blocks to find the general solution to the defect Casimir equations \cite{Billo:2016cpy}, i.e. \eq{general-casimir-soln}.

\subsection{Positivity in reciprocal NN-FTs}

In ordinary Euclidean CFTs, reflection positivity is formulated by requiring that two-point functions of a field and its adjoint define a non-negative quadratic form for any finite linear combination of insertions. In our NN–FT setting, we adopt the following minimal analogue: a two-point kernel $G_\Delta(X_1,X_2)$ is called \emph{positive} if, for any finite set of Euclidean points $X_i$ and complex coefficients $c_i$,
\begin{equation}
    \sum_{i,j}\bar c_i c_j\,G_\Delta(X_i,X_j)\;\ge 0\,.
\end{equation}
We will refer to this as `positivity' of the two-point function by abuse of language, without constructing the full Osterwalder–Schrader framework.

In the NN-FT paradigm, let $P(\Theta)$ be a positive probability measure on parameter space, let $\CO_\Delta(X;\Theta)$ be a (complex-valued) functional of $\Theta$ at a point $X$ in the input space. Assuming there is an antilinear involution acting as conjugation on fields, i.e. $\mathcal R:\CO_\Delta(X;\Theta)\mapsto\CO_\Delta(X;\Theta^\dagger)
    = \overline{\CO_\Delta(X;\Theta)}$, we see that the NN-FT two-point function $G_\Delta(X_i,X_j) = \mathbb{E}_{P(\Theta)}\big[\CO_\Delta(X_i;\Theta)\,\overline{\CO_\Delta(X_j;\Theta)}\big]$ satisfies
\begin{equation}\label{eq:reflection-positivity}
    \sum_{i,j}\bar c_i c_j\,G_\Delta(X_i,X_j)
    = \mathbb{E}_{P(\Theta)}\left[\Big|\sum_i c_i\,\CO_\Delta(X_i;\Theta)\Big|^2\right]\geq 0\,,
\end{equation}
which implies that, whenever the expectations are finite, the kernel $G_\Delta$ is positive in the above sense. In the remainder of this subsection, we determine, for the reciprocal theories of interest and for specific choices of priors $P(\Theta)$, the range of $\Delta$ for which this probabilistic representation is valid, and how the two-point functions analytically continue beyond that range once the bare $\Theta$-integrals cease to converge.

\subsubsection*{Real, projective prior}

We first keep $\Theta$ real and relax isotropy. Take a projective, homogeneous prior and define
\begin{equation}
    \mathrm{d}P_{\rm hom}(\Theta)
    = f(\rho)\,w(\check\Theta)\,\mathrm{d}\rho\,\mathrm{d}\check\Theta\,, \qquad p(\Theta) := f(|\Theta|)\,w(\check\Theta)~,
\end{equation}
with $\rho = |\Theta|$, $\check\Theta\in \mathds{S}^{D-1}$, and $f(\rho)\geq0$, and where $w(\check\Theta)\geq0$ is normalized to unity with the assumption that $p$ is smooth and strictly positive in a neighborhood of $\Theta=0$.

Using rotational symmetry, we map $X\to |X|e_1$, so that $\Phi_{\Delta}(X;\Theta) = (|X|\Theta_1)^{-\Delta}$. The diagonal two-point function is
\begin{equation}
    \mathbb{E}_{P_{\rm hom}}\big[|\Phi_\Delta(X;\Theta)|^2\big]
    = |X|^{-2\Delta}\int_{\mathbb R^D}\mathrm{d}^D\Theta\,p(\Theta)\,|\Theta_1|^{-2\Delta}~.
\end{equation}
Near $\Theta_1=0$, the density is bounded and nonzero, so the $\Theta_1$-integral converges in the interval $0\le\Delta<\tfrac{1}{2}$. In this window, $\Phi_\Delta(X)$ is an $L^2$ random variable and $G_\Delta$ is a genuine covariance, so \eqref{eq:reflection-positivity} applies.

To go beyond $0\le\Delta<\tfrac{1}{2}$ we analytically continue the $\Theta$-integral. Fix $X\neq0$ and introduce a thin slab around the surface $\Theta_1=0$,
\begin{equation}
    H(\epsilon) := \big\{\,|\Theta_1|<\epsilon,\ \Theta_\perp\in B\,\big\}~,
\end{equation}
where $\Theta_\perp=(\Theta_2,\dots,\Theta_D)$ and $B$ is a bounded transverse region of volume $V_{D-1}$. The singular contribution to the diagonal two-point function is
\begin{equation}
    \mathcal{I}(\Delta;\epsilon)
    = \int_{H(\epsilon)}\mathrm{d}^D\Theta\,p(\Theta)\,|\Theta_1|^{-2\Delta}\simeq p(0)\,V_{D-1}\,\frac{2\,\epsilon^{1-2\Delta}}{1-2\Delta}+\mathcal O(\epsilon^s)~,
\end{equation}
where in the last equality we expanded in $p(\Theta)=p(0)+\mathcal O(|\Theta|)$ to leading order, and where $s>0$. Clearly, $G_\Delta$ has a simple pole at $\Delta=\tfrac{1}{2}$ with a logarithmic singularity as $\epsilon\to0$. Expanding near the first pole at $\Delta_*=\tfrac{1}{2}$ by writing $\Delta=\Delta_*+\delta$, we find
\begin{equation}
    G_\Delta^{(\epsilon)}(X,X)
    = |X|^{-2\Delta_*}\Big[
        \frac{\tilde{\mathsf A}(X)}{\Delta-\Delta_*}
      + \tilde{\mathsf B}(X)\,\log\frac{1}{\epsilon}
      + \tilde{\mathsf F}_\Delta(X)
      + \mathcal O(\epsilon^s)
    \Big]~,
\end{equation}
with $\tilde{\mathsf F}_\Delta(X)$ holomorphic in $\Delta$ near $\Delta_*$, and $\tilde{\mathsf B}(X)\;\propto\;p(0)\,V_{D-1}$
independently of the detailed implementation of the cutoff. In particular, using the above scheme, the renormalized diagonal two-point function is $G_{\Delta_*}^{\rm ren}(X,X)
    = \tilde{\mathsf B}(X)\,|X|^{-2\Delta_*}$. 

For generic $\Delta$, one defines $G_\Delta^{\rm ren}(X_1,X_2)$ by analytic continuation of the regulated integral. Keeping higher terms in the expansion of $p(\Theta_1,\Theta_\perp)$ along $\Theta_1$,
\begin{equation}
    p(\Theta_1,\Theta_\perp)
    = \sum_{k\ge0} b_k(\Theta_\perp)\,\Theta_1^k
    = \sum_{m\ge0} a_m(\Theta_\perp)\,\Theta_1^{2m}\,,\qquad a_m:=b_{2m}~,
\end{equation}
and inserting this into the local analysis, one finds that each even-$m$ term near $\Delta_m = \frac{1}{2}+m$ contributes a singular piece of the form
\begin{equation}
    \mathcal I_m(\Delta;\epsilon)
    \sim A_m\,\frac{\epsilon^{2m+1-2\Delta}}{2m+1-2\Delta}~,
\end{equation}
with $A_m = \int_B\mathrm{d}^{D-1}\Theta_\perp\,a_m(\Theta_\perp)$. Thus, $G_\Delta$ has an infinite tower of simple poles at $\Delta=\Delta_m$, each accompanied by a universal $\log(1/\epsilon)$ term, whose coefficient is fixed by the corresponding transverse moment $A_m$ of the prior. 

\subsubsection*{Complex isotropic prior}

Keeping real inputs $X\in \mathbb{R}^d$, we now allow complex-valued parameters. Let $\Theta=(\widehat\Theta,\widetilde\Theta)\in\mathbb{C}^D$ split into tangential $\widehat\Theta\in\mathbb{C}^{p+2}$ and normal $\widetilde\Theta\in\mathbb{C}^q$ components. We choose a complex isotropic Gaussian prior
\begin{equation}
    P(\Theta)\propto
    \exp\!\Big(-\frac{\Vert\widehat\Theta\Vert_\bullet^2}{\widehat\sigma^2}\Big)\,
    \exp\!\Big(-\frac{\Vert\widetilde\Theta\Vert_\circ^2}{\widetilde\sigma^2}\Big)~,
\end{equation}
where, for complex arguments, we define $\Vert\widehat\Theta\Vert_\bullet^2 := \widehat\Theta^\dagger\bullet\widehat\Theta$ and $\Vert\widetilde\Theta\Vert_\circ^2 := \widetilde\Theta^\dagger\circ\widetilde\Theta$. We also define the basic complex Gaussian variables
\begin{equation}
    \widehat{\mathcal{Z}}(X) := \widehat\Theta\bullet X\,,\qquad
    {\mathcal Z}(X) := \Theta\cdot X = \widehat\Theta\bullet X + \widetilde\Theta\circ X~,
\end{equation}
with covariance
\begin{equation}
    \Sigma(X_1,X_2)
    := \mathbb{E}\big[\mathcal{Z}(X_1)\,\overline{\mathcal{Z}(X_2)}\big]
    = \widehat{\sigma}^2\,(X_1\bullet X_2) + \widetilde{\sigma}^2\,(X_1\circ X_2)~.
\end{equation}
The ambient reciprocal fields are $\Phi_\Delta(X) =\mathcal Z(X)^{-\Delta}$, while the defect fields are $
    \widehat\varphi_\Delta(X) :=\widehat{\mathcal Z}(X)^{-\Delta}$. Lastly, we define adjoints by
\begin{equation}
    \Phi_\Delta^\dagger(X;\Theta) := \overline{\Phi_\Delta(X;\Theta)}=(\Theta^\dagger\cdot X)^{-\Delta}\,,\qquad
    \widehat\varphi_\Delta^\dagger(X;\Theta) := \overline{\widehat\varphi_\Delta(X;\Theta)}=(\widehat\Theta^\dagger\bullet X)^{-\Delta}~.
\end{equation}

The key input is the negative moment of a complex Gaussian distribution. On the (ambient, resp. defect) Poincar\'e section, $\mathcal{Z}(X)$ (resp. $\widehat{\mathcal Z}(X)$) has a complex Gaussian distribution with variance $\Sigma(X,X)$ (resp. $\widehat\Sigma (X,X)= \Sigma(X,X)|_{\tilde\sigma\to 0}$), so for $0\le\Delta<1$,
\begin{equation}
\big\langle\Phi_\Delta(X_1)\,\Phi_\Delta^\dagger(X_2)\big\rangle
    = \Gamma(1-\Delta)\,\big[\Sigma(X_1,X_2)\big]^{-\Delta}~,
\end{equation}
with $\big\langle\widehat\varphi_\Delta(X)\,\widehat\varphi_\Delta^\dagger(X)\big\rangle$ following similarly. In the entire strip $0\le\Delta<1$, these correlators are expectations over the Gaussian ensemble $P(\Theta)$, so $G_\Delta$ is positive semidefinite. 

For $\Delta\ge1$, the negative moments diverge, and the expression with $\Gamma(1-\Delta)$ must be understood as an analytic continuation. The two-point function becomes a meromorphic function of $\Delta$ with simple poles at positive integers, with the  expansion near $\Delta=j\in\mathbb N$,
\begin{equation}
    \Gamma(1-\Delta)
    \simeq \frac{(-1)^j}{(j-1)!(\Delta-j)} +\mathcal O(1)~.
\end{equation}
 On each interval $(j,j+1)$, the sign of $\Gamma(1-\Delta)$ is fixed and alternates across poles, but in this regime the correlator is no longer a covariance with respect to the fixed Gaussian prior, and hence global positivity is not guaranteed.

The hard-cutoff analysis again makes the tower of poles and logarithmic divergences explicit. The dangerous locus is the complex line $\mathcal Z(X)=\Theta\cdot X=0$, and the regulated diagonal kernel
\begin{equation}
    G_\Delta^{(\epsilon)}(X,X)
    := \int_{\{|\mathcal Z(X)|\ge\epsilon\}}\mathrm d\Theta\,
          P(\Theta)\,|\mathcal Z(X)|^{-2\Delta}
\end{equation}
is controlled near $|\mathcal Z(X)|=0$ by the small-$r$ behavior of the effective radial density $\tilde p(r^2)$ for $\mathcal Z(X)=re^{i\varphi}$.  Expanding $\tilde p(r^2)
    = \tilde p_0 + \tilde p_1 r^2 + \tilde p_2 r^4 + \ldots$, we find that the singular part of $G_\Delta$ behaves schematically as
\begin{equation}
    G_\Delta^{(\epsilon)}(X,X)
    \simeq \sum_{n\ge0}
    \tilde C_n(X)\,\frac{\epsilon^{\,2(n+1-\Delta)}}{2(n+1-\Delta)}
    + \text{(holomorphic in $\Delta$)}~,
\end{equation}
with $\tilde C_n(X)\propto\tilde p_n$. There is an infinite tower of simple poles at $\Delta_n = 1+n\in \mathbb{N}$, and expanding near each $\Delta_n$ as $\Delta=\Delta_n+\delta$ gives
\begin{equation}
    \frac{\epsilon^{\,2(n+1-\Delta)}}{2(n+1-\Delta)}\Big|_{\delta\to 0^+}
    = -\frac{1}{2\delta} - \log\frac{1}{\epsilon} + \mathcal O(\delta)~.
\end{equation}
Thus every pole at $\Delta=1+n$ is accompanied by a universal $\log(1/\epsilon)$ term whose coefficient is fixed by $\tilde p_n$. The leading pole at $n=0$ corresponds to the boundary of the convergence strip at $\Delta_\ast=1$, where the prefactor $\Gamma(1-\Delta)$ develops its first pole.

Taken together, these observations highlight two notions of ``existence'' for reciprocal NN-FT correlators. For $\Delta$ in the convergence window determined by the prior ($0\le\Delta<\tfrac{1}{2}$ for real homogeneous priors, $0\le\Delta<1$ for complex Gaussians), the two-point functions arise as genuine covariances of $L^2$ random fields on parameter space, and satisfy the positivity condition \eqref{eq:reflection-positivity}. Beyond this window, the bare $\Theta$-integrals diverge, but the hard-cutoff analysis shows that they admit a canonical analytic continuation in $\Delta$, organized as a meromorphic family with a tower of simple poles and universal logarithmic divergences at $\Delta=\Delta_\ast+j$ ($\Delta_\ast=\tfrac{1}{2}$ or $1$ depending on the prior). The coefficients of these logarithms are local functionals of the prior near the singular locus and are insensitive to the detailed shape of the regulator. In this renormalized sense, reciprocal NN-FTs provide well-defined two-point functions for arbitrary $\Delta$: the tower of logarithmic coefficients encodes universal, locally determined information about the prior near the singular locus and controls how one interpolates between the strictly probabilistic neural-ensemble regime and the purely meromorphic continuation of the theory.

\section{Conclusion}\label{sec:conclusions}

In the present work, we extended the formalism in \cite{Halverson:2024axc}, which allows for conformal symmetries to be realized in NN-FTs, to include deformations by extended objects --- defects and boundaries. By modifying the embedding space approach used to build conformal fields from NNs to account for the defect subgroup of the global conformal symmetry group, including the normal bundle symmetries, we have expanded the scope of this novel paradigm to define field theories with non-local insertions of arbitrary codimension. These conformal defects alter the constraints that conformal symmetries place on correlation functions, allowing for non-trivial one-point functions for ambient fields to develop.  The non-Gaussianities arising as ambient-to-defect couplings are novel in NN-FTs. Moreover, we have proposed a formalism, based on OPE-like expansions of ambient fields in the presence of a boundary or defect, that decomposes ambient fields into irreducible representations of the defect symmetries. This has enabled a systematic approach to the computation of correlation functions involving ambient fields in terms of defect conformal blocks. The upshot of this OPE-like picture for defects is that it naturally has in its definition a notion of transverse `spin', which is the quantum number for the normal bundle symmetry; this has led to developments in the direction of further generalizing the NN-FTs to include conformal fields carrying non-trivial spacetime spin.

In the sections following the set up of the conformal defect formalism, we have put these new tools to the test with some simple toy example NN-FTs.  In both examples considered, we started from Gaussian-distributed random variables, e.g. by taking the distribution of NN parameters to be $P(\Theta) = \CN(\Theta\mid0,\mathds{1}_d)$. The first example that we studied was the class of `non-unitary' theories $\Phi_{\Delta}(X)= (\Theta\cdot X)^{-\Delta}$, encountered in \cite{Halverson:2024axc}, with scaling dimension $\Delta<0$. Since the basic architecture of the network has positive powers of $\Theta\cdot X$, we referred to these as monomial theories.  In this case, we were able to flesh out details of how to use the new defect formalism and compute one- and two-point functions for various combinations of defect and ambient fields.  Using the OPE-like expansion, we were then able to exactly solve the $SO(p+1,1)$ and $SO(q)_N$ Casimir equations, and thus explicitly find closed-form expressions for the defect conformal blocks in these unusual conformal theories.

In the second set of examples we considered, we used the same architecture, albeit with $\Delta >0$.  This posed several interesting challenges, in that there are obvious pathologies in correlation functions around the origin $\Theta\cdot X =0$ that need to be regulated. In the main case we considered, namely that of the real isotropic Gaussian prior distribution for the network parameters, we developed a regulating procedure by first passing to Feynman parameterization and introducing a cutoff at the endpoint of the intervals where the resulting integrals diverged. We then observed that the coefficient of the log-divergent term of the cut-off integral, which is universal and independent of the shape of the regulator, was precisely of the form of a two-point function of conformal scalar fields.  We then used the defect-defect version of this regulated two-point function together with the OPE-like expansions of ambient reciprocal fields to build ambient-defect and ambient-ambient two-point functions.

Our analysis of the two-point functions in reciprocal theories was not free of surprises.  In developing the regulating procedure, we utilized a series expansion near the defect locus in terms of a tower of defect and normal fields, which naturally mimicked the OPE-like expansions we previously used for monomial theories.  The way that the expansion presented itself, though, was somewhat unconventional, in that the defect identity operator was not present.  This had drastic implications for the one-point function of ambient fields, which --- as the only coupling that can appear in the defect OPE is that to the defect identity --- vanished identically in our regulating scheme. 

Perhaps more troubling was an observed behavior that we subsequently heavily scrutinized, namely that of positivity selection rules in the space of fields, where, for certain values of scaling dimension, the two-point function was strictly negative.  We saw upon further investigation that the selection rules to maintain positivity in two-point functions are a generic feature of reciprocal theories. That is, allowing for real projective (homogeneous) but non-isotropic as well as complex isotropic Gaussian prior distributions, we observed that continuation as a meromorphic function from the domain of absolute convergence of the two-point function through simple poles at the upper bound in the space of $\Delta$ naturally led to a loss of positivity globally.

\subsection*{Future Directions}

\begin{itemize}
    \item {\textbf{Global symmetries and monodromy defects.}} As briefly touched upon in \sn{NN-CFT}, implementing global symmetries in a field-theoretic way in NN-FTs requires a novel approach.  That is, we would naturally draw the parallel from global symmetry in field space implemented by acting on a collection of fields with some group action where the action functional and measure in the path integral, if the symmetry is non-anomalous, remain invariant to a NN-FT where the group acts on an ensemble of networks. In order for the NN-FT to be invariant under the symmetry, we would need to require that the generalized measure in the ensemble $\int D\Theta P(\Theta)$ and the expectations computed with respect to it remain invariant. 

    This opens up an interesting question in the context of defects in NN-FTs.  There are a class of co-dimension $q=2$ defects in theories with global symmetries that are in a sense `trivial' \cite{Lauria:2020emq} in that they can only be defined with reference to an ambient field and are constructed as boundaries for co-dimension $q=1$ topological operators that implement a symmetry transformation, i.e. monodromy defects.  Despite their `trivial' construction they exhibit rich defect physics, see e.g. \cite{Bianchi:2021snj}.  In particular, \cite{Bianchi:2021snj} showed that monodromy defects in free fermion CFTs require certain defect modes in order to be consistent with ambient equations of motion, and given the work in \cite{Frank:2025zuk} constructing fermionic NN-FTs, it would be interesting to consider how such monodromy defects can be constructed in NN-FTs and what novel physics can be inferred from them.
    
    \item {\textbf{Spinning conformal fields.}} As mentioned above, the OPE-like expansion that we use to decompose ambient insertions in terms of the irreps of the defect conformal symmetry and the normal bundle symmetry naturally suggests an extension of the conformal field from NN formalism to include `spinning conformal fields'.  In our case, the transverse spin was simply the quantum number of the global $SO(q)_N$ symmetry.  In general, we should employ a set of auxiliary null coordinates that can be used to absorb free indices at the price of having to encode additional geometric structure in the embedding space formalism; see the brief review of this approach for symmetric traceless tensors (STTs) in \cite[Sec. 2]{Billo:2016cpy}.  For STTs $\Phi^{M_1\ldots M_n}(X)$, the extension should follow easily as the additional constraint on correlations functions of transversality $X_M \Phi^{M\ldots} = 0$ is enforced by introducing $n$ new coordinates $Z$ satisfying $Z^2=Z\cdot X = 0$. We can then use the usual embedding space technology for $\Phi_n = Z_{M_1}\cdots Z_{M_n}\Phi^{M_1\ldots M_n}(X)$, where the indices can be freed afterward by the action of the Todorov operator.  The construction of the formalism to encode non-trivial spacetime spin in conformal fields from NNs, and then extending to spinning defects, is the subject of on-going work by the authors.
    
    \item {\textbf{Anomalies.}} Much like other symmetries, conformal symmetries can be anomalous, and indeed their anomalies are of great interest as they are robust characteristics of a fixed point in the space of QFTs.  However, conformal anomalies are seen when the theory is placed on a curved background.  An interesting line of inquiry thus opens up: If the space from which the input data to the NN has non-trivial curvature, i.e. in the FT picture the background geometry is curved, we can then use usual FT techniques to compute the anomalies, e.g. through correlation functions of the stress tensor. The challenge then is to interpret the meaning of these physical quantities in the properties of NNs.  Including conformal defects in this program would further enrich the story as the defects have their own conformal anomalies that can have wildly different structures as compared to the ambient theory, e.g. \cite{Chalabi:2021jud}.
    
    \item {\textbf{Interacting conformal fields in higher dimensions.}} The structure of the ambient theories that we deform by the inclusion of defects in this paper are (generalized) free field theories.  From dimensional analysis, it is expected that there are no non-trivial (i.e. interacting) fixed points in higher dimensions ($d>6$), or that the structure of such a theory is beyond the capacity to be understood in a Lagrangian formalism.  However, the conformal fields that we have seen being constructed from NNs exist in arbitrary dimension, only constrained by the dimensionality of the input space. Introducing non-Gaussianities in parameter space can be understood as turning on interactions for the conformal fields.  Thus, it seems as though the NN-FT has the potential capacity to build examples of higher-dimensional interacting conformal fields.  It would be tremendously interesting to explore the ways in which the NN-FT paradigm can evade these constraints.

    \item \textbf{Two-dimensional theories.}  Given the flexibility of conformal fields in NN-FTs in that they can be defined in arbitrary dimensions, it is natural to ask how to construct novel two-dimensional conformal fields and whether they can realize Virasoro symmetry.  We will address this problem in upcoming work \cite{Robinson:2025ybg}; with an eye toward extending our defect construction to explore boundary physics in two-dimensions, e.g. Ishibashi states, from the perspective of NNs.

\end{itemize}

\section*{Acknowledgments}
We would like to thank Jim Halverson for contributions and discussion in the early phase of this project.  We would also like to thank Nathan Benjamin, Christian Ferko, Silviu Pufu, and Justinas Rumbutis for illuminating discussions during this work.  The work of P.C. is supported by a Mayflower studentship from the University of Southampton.  The work of B.R. is supported by an NWO vidi grant (number 016.Vidi.189.182). The work of B.S. is partly supported by the National Research Foundation of Korea (NRF) grant RS-2025-00516583 and by the STFC consolidated grant ST/T000775/1.

\appendix
\section{Useful identities}\label{app:two-point-fns}

\begin{lemma}\label{lemma:1ptGaussianDefect}
    Let $\mathcal{N}(0,\hat\sigma^2\mathds{1}_{d+2})$ be a multivariate Gaussian distribution centered at the origin with variance matrix $\hat\Sigma_{AB}=\hat\sigma^2\delta_{AB}$. The expectation value of the defect variables $\hat\theta$ then obeys
    \begin{align}
        \mathbb{E}[(X\bullet\theta)^{2\lambda}] &= Q(2\lambda)(X\bullet X)^{\lambda}\hat\sigma^{2\lambda} \overset{\mathrm{P.S.}}{=}0~,
    \end{align}
    for all $\lambda\in\mathbb{N}^*$. The function $Q(n)$ is defined as counting all possible pairings between $n$ (even) objects
    \begin{align}
    Q(n) &= 2^{n/2}\pi^{-1/2}\Gamma\left(\frac{n+1}{2}\right)~.
    \end{align}
\end{lemma}

\begin{proof}
    Let us prove this statement recursively.
    
    When $\lambda=1$, we explicitly find $\mathbb{E}[(X\bullet\theta)^2]=\Vert X\Vert_\bullet^2 \,\hat\sigma^2$. Let us now assume that $\mathbb{E}[(X\bullet\theta)^{2\lambda}] = \frac{(2\lambda-1)!}{2^{\lambda-1}(\lambda-1)!}\Vert X\Vert_\bullet ^{2\lambda}\hat\sigma^{2\lambda}$ for some $\lambda\in\mathbb{N}^*$. Then, using Stein's lemma (with $\hat\mu_1=0$ and $\hat\Sigma_{AB}=\hat\sigma^2\delta_{AB}$)~,
    \begin{align}
        \mathbb{E}[g(\theta)\theta_A] &= \hat\sigma^2\sum_{B}\delta_{AB}\mathbb{E}[\partial_B g(\theta)]~,
    \end{align}
    we find the intermediate result
    \begin{align}
        \mathbb{E}[\theta_{A_1}\cdots\theta_{A_{2\lambda}}\theta_{B_1}\theta_{B_2}] &= \hat\sigma^2\sum_{i=1}^{2\lambda}\delta_{B_2 A_i}\mathbb{E}[\theta_{A_1}\cdots\widehat{\theta_{A_i}}\cdots\theta_{A_{2\lambda}}\theta_{B_1}]+\hat\sigma^2\delta_{B_1B_2}\mathbb{E}[\theta_{A_1}\cdots\theta_{A_{2\lambda}}]~.
    \end{align}
    Contracting this with the vector components of $X$, we are able to evaluate the expression at $\lambda+1$,
    \begin{align}
        \mathbb{E}[(X\bullet\theta)^{2\lambda+2}] &= X^{A_1}\cdots X^{A_{2\lambda}}X^{B_1}X^{B_2}\mathbb{E}[\theta_{A_1}\cdots\theta_{A_{2\lambda}}\theta_{B_1}\theta_{B_2}]\nonumber\\
        &=(2\lambda+1)\Vert X\Vert_\bullet^2\hat\sigma^2\mathbb{E}[\Vert X\Vert_\bullet ^{4\lambda}]\nonumber\\
        &= \frac{(2\lambda+1)!}{2^{\lambda}\lambda!}\Vert X\Vert_\bullet ^{2\lambda+2}\hat\sigma^{2(\lambda+1)}~.
    \end{align}
\end{proof}

\begin{lemma}\label{lemma:2ptGaussianDefect}
    Let $\mathcal{N}(0,\hat\sigma^2\mathds{1}_{d+2})$ be a multivariate Gaussian distribution centered at the origin with variance matrix $\hat\Sigma_{AB}=\hat\sigma^2\delta_{AB}$. Let $X_1$, $X_2$ be two distinct $(d+2)$-dimensional vectors. Then the expectation value of the random variables $\hat\theta$ obeys
    \begin{align}\label{eq:lemma-2}
        \mathbb{E}[(X_1\bullet\theta)^{n_1}(X_2\bullet\theta)^{n_2}] = \sum_{\substack{
        n=0\\
        n_1-n=0\,[2]\\
        n_2-n=0\,[2]}}^{\operatorname{min}(n_1,n_2)}&(X_1\bullet X_2)^n\Vert X_1
        \Vert_\bullet^{n_1-n} \Vert X_2\Vert_\bullet^{n_2-n}\hat\sigma^{n_1+n_2}\times\\\nonumber &\hspace{0.25cm}\times\frac{2^{\frac{2n-n_1-n_2}{2}}\Gamma(n_1+1)\Gamma(n_2+1)}{\Gamma(n+1)\Gamma(\frac{n_1-n+2}{2})\Gamma(\frac{n_2-n+2}{2})}
    \end{align}
    for all $n_1,n_2\in\mathbb{N}$. A similar expression holds for the tangent space random variables $X_1\circ\theta$, where $\bullet$ should be replaced by $\circ$ and $\hat{\sigma}^2$ by $\tilde{\sigma}^2$.

    For defect variables, all terms proportional to $X_1\bullet X_1$ and $X_2\bullet X_2$ vanish on the defect Poincar\'e section, and so the expectation value simplifies to
    \begin{align}
        \mathbb{E}[(X_1\bullet\theta)^{n_1}(X_2\bullet\theta)^{n_2}] &\overset{\mathrm{P.S.}}{=} \delta_{n_1 n_2}\Gamma(n_1+1)(X_1\bullet X_2)^{n_1}\hat\sigma^{2n_1}~.
    \end{align}
\end{lemma}

\begin{proof}
    The expectation value is given by all possible contractions between $n_1$ copies of the vector $X_1$ and $n_2$ copies of the vector $X_2$. We can classify them according to the number of contractions $X_1\bullet X_2$, call it $n$.

    For a given $n$, there are $\frac{1}{n!}P(n_1,n)P(n_2,n)$ ways of constructing these $n$ contractions of $X_1$ and  $X_2$. $P(m,n)$ denotes the number of permutations of $n$ objects within a set of $m$ objects, and is defined as $P(m,n)=\frac{m!}{(m-n)!}$. The remaining $n_1-n$ copies of $X_1$ must then be contracted within themselves. Similarly for the $n_2-n$ copies of $X_2$. Of course, this is only possible if both $n_1-n$ and $n_2-n$ are even. Whenever that is the case, there are $Q(n_1-n)Q(n_2-n)$ ways of contracting these vectors among themselves, as shown in lemma~\ref{lemma:1ptGaussianDefect}.
\end{proof}

Note that \eq{lemma-2} can be resummed in terms of a hypergeometric function.
    \begin{subequations}
    For $n_1$, $n_2$ both even, this gives
    \begin{align}
        \mathbb{E}[(X_1\bullet\theta)^{n_1}(X_2\bullet\theta)^{n_2}] &= \Vert X_1\Vert_\bullet^{n_1}\Vert X_2\Vert_\bullet^{n_2}\left(\frac{\hat\sigma^2}{2}\right)^{\frac{n_1+n_2}{2}}\frac{\Gamma(n_1+1)\Gamma(n_2+1)}{\Gamma\left(\frac{n_1}{2}+1\right)\Gamma\left(\frac{n_2}{2}+1\right)}\times\nonumber\\
        &\qquad\times{}_2F_1\left(-\frac{n_1}{2},-\frac{n_2}{2};\frac{1}{2};\frac{(X_1\bullet X_2)^2}{\Vert X_1\Vert_\bullet^2 \Vert X_2\Vert_\bullet^2}\right)~,
    \end{align}
    whilst for $n_1$, $n_2$ both odd,
    \begin{align}
        \mathbb{E}[(X_1\bullet\theta)^{n_1}(X_2\bullet\theta)^{n_2}] &= \frac{X_1\bullet X_2}{\Vert X_1\Vert_\bullet^{1-n_1}\Vert X_2\Vert_\bullet^{1-n_2}}\hat\sigma^{n_1+n_2}\frac{2^{\frac{2-n_1-n_2}{2}}\Gamma(n_1+1)\Gamma(n_2+1)}{\Gamma\left(\frac{n_1-1}{2}+1\right)\Gamma\left(\frac{n_2-1}{2}+1\right)}\times\\\nonumber
        &\qquad\times{}_2F_1\left(\frac{1-n_1}{2},\frac{1-n_2}{2};\frac{3}{2};\frac{(X_1\bullet X_2)^2}{\Vert X_1\Vert_\bullet^2 \Vert X_2\Vert_\bullet^2}\right)~,
    \end{align}
    \end{subequations}
    and zero otherwise.
\bibliographystyle{JHEP}
\bibliography{DefectNNs}
\end{document}